\documentclass[letterpaper,12pt,margin=0.5in]{article}
\usepackage[letterpaper, top=1.6in, bottom=1.6in,left=1in,right=1in]{geometry}
\usepackage[utf8]{inputenc}
\usepackage{amsmath}
\usepackage{amssymb,xcolor,verbatim}
\usepackage{hyperref}
\usepackage{amsthm}
\usepackage{tikz}
\usepackage{subcaption}
\usepackage{cleveref}
\usetikzlibrary{arrows}
\usepackage{forest}
\usepackage{graphicx}
\usepackage{enumitem}

\usepackage[
style=authoryear,
]{biblatex}
\usetikzlibrary{patterns}
\def\BState{\State\hskip-\ALG@thistlm}

\newcommand{\drightarrow}{\overset{d}{\rightarrow}}
\newcommand{\supp}[2][]{\text{supp}_{#1}(#2)}
\newcommand{\supps}[1]{\text{supp}_{\Theta}(#1)}
\newcommand{\CM}{\text{CM}(\sigma)}
\newcommand{\GCM}{\text{GM}(\sigma)}
\newcommand{\SLID}[1][\sigma]{\text{SGM}(#1)}
\newcommand{\SLIDS}[1][\sigma]{\text{SLID}(#1)}
\newcommand{\inte}[1]{\text{int}\left(#1\right)}

\newcommand{\E}[2][]{\mathbb{E}_{#1}\left[#2\right]}
\renewcommand{\P}[2][]{\mathbb{P}_{#1}\left[#2\right]}
\newcommand{\Econd}[3][]{\mathbb{E}_{#1}\left[#2\middle\rvert #3\right]}
\newcommand{\Pcond}[3][]{\mathbb{P}_{#1}\left[#2\middle\rvert #3\right]}

\newcommand\blfootnote[1]{%
  \begingroup
  \renewcommand\thefootnote{}\footnote{#1}%
  \addtocounter{footnote}{-1}%
  \endgroup
}
\newtheorem{lem}{Lemma}
\newtheorem{thm}{Theorem}

\newtheorem{prop}{Proposition}
\newtheorem{cor}{Corollary}
\newtheorem{manualtheoreminner}{Theorem}
\newenvironment{manualtheorem}[1]{%
  \renewcommand\themanualtheoreminner{#1}%
  \manualtheoreminner
}{\endmanualtheoreminner}
\newtheorem*{fact}{Observation}
\newtheoremstyle{named}{}{}{\itshape}{}{\bfseries}{.}{.5em}{\thmnote{#3}}
\theoremstyle{named}
\newtheorem*{axiom}{Assumption}
  
\theoremstyle{definition}
\newtheorem{defn}{Definition}

\newtheorem{ex}{Example}

\makeatother
\newcounter{parentnumber}
\usepackage{newtxtext,newtxmath}
\DeclareSymbolFont{CMletters}{OML}{cmm}{m}{it}
\DeclareMathSymbol{v}{\mathord}{CMletters}{`v}
\DeclareMathOperator*{\argmax}{arg\,max}
\DeclareMathOperator*{\argmin}{arg\,min}

\title{Robust Communication Between Parties with \emph{Nearly} Independent Preferences
\blfootnote{I would like to thank Dilip Abreu for feedback and guidance on this project, I am grateful to Elliot Lipnowski, Joyee Deb, Basil Williams, Doron Ravid, Erik Madsen, and Debraj Ray for helpful comments.}}
\date{First Version: August 16, 2022\\
Updated: \today
}
\author{Alistair Barton}
\begin{document}
\maketitle
\begin{abstract}
    We study finite-state communication games in which the sender's preference is perturbed by random private idiosyncrasies. Persuasion is generically impossible within the class of statistically independent sender/receiver preferences --- contrary to prior research establishing persuasive equilibria when the sender's preference is precisely transparent.
    
    Nevertheless, robust persuasion may occur when the sender's preference is only slightly state-dependent/idiosyncratic. This requires approximating an `acyclic' equilibrium of the transparent preference game, generically implying that this equilibrium is also `connected' --- a generalization of partial-pooling equilibria. It is then necessary and sufficient that the sender's preference satisfy a monotonicity condition relative to the approximated equilibrium.
    
    If the sender's preference further satisfies a `semi-local' version of increasing differences, then this analysis extends to sender preferences that rank pure actions (but not mixed actions) according to a state-independent order.
    
    We apply these techniques to study (1) how ethical considerations, such as empathy for the receiver, may improve or impede communication, and (2) when the sender may benefit from burning money as a persuasion tactic.  
\end{abstract}
Classic models of communication rely on the essential property that the sender's preference is strongly affected by their information. In cheap talk games, knowledge of the state typically leads to significantly different preferences over the receiver's action; in signalling games the state affects the sender's preference over the signal sent. 

A recent pattern of literature has focussed on the intriguing possibility of communication when the sender's preference is unaffected by their information (\cite{CH10,LR20}). However such models strictly rely on the sender's preference being \emph{transparent}, ie. both \emph{public} and state-independent. If the assumption of public preferences is relaxed, to study the broader set of independent preferences --- with sender idiosyncrasy --- we observe that communication is generically impossible, fundamentally because it relies on weakly dominated strategies. If we relax the assumption of state-independence it is not obvious whether communication will be preserved (and, if so, which communication schemes will survive). It may seem intuitive that aligning preferences or making the sender lying-averse would preserve communication; we will see that this is not always the case.

Our aim in this paper is to study how robust communication may occur in the space between the classic models, with their large state-dependence, and these recent models with no state-dependence, and, in particular, the role played by (limited) state-dependence in the sender's preference in shaping this communication. Our results characterize equilibria featuring the robustness of the classic models while only requiring slight state-dependence in the sender's preference.

In our model, the sender's preference $u_S$ is composed of a transparent utility $u^0$ and a state-dependent, possibly idiosyncratic, perturbation $V$:
\begin{equation}\label{eq:intro}
    u_S(a,m\rvert \theta,\omega)=u^0(a,m)+\epsilon V(a,m\rvert \theta,\omega)
\end{equation}
where $a$ is the receiver's action, $m$ is the sender's message\footnote{In cheap talk cases, message dependence may be omitted from the utility.}, $\theta$ is the state, and $\omega$ is the realized idiosycrasy. We search for equilibria that are robust to slight variations in the perturbation, constraining our model to finite states --- allowing a tractable definition of robustness --- and finite actions --- providing `regularity' in player utilities.

Upper-hemicontinuity of equilibria as $\epsilon\rightarrow 0$ ensures that the equilibria of our model will approximate equilibria of the transparent preference model (which we term \emph{candidate equilibria}). Such an equilibrium is robust to idiosyncrasy if \emph{for any} $V$, we obtain approximating equilibria as $\epsilon\rightarrow 0$. No informative candidate equilibrium satisfies this robustness check. We relax this criteria by constraining $V$ to be supported within some open set $\mathscr{O}$. If an equilibrium is approximated by such perturbations, we say it is \emph{$\mathscr{O}$-stable}.

$\mathscr{O}$-stability describes how robust equilibria may occur when the sender's idiosyncrasy is overwhelmed by some state dependent modification to their utility. This modification may describe the sender's ethical concerns about misleading the receiver, or some second-order material/technical cost to the sender (a well-known example is the Beer-Quiche signalling model of \cite{IK87}).

The main result of this paper is a precise characterization of $\mathscr{O}$-stability. A key element of our analysis is \emph{communication graphs}, these are bipartite graphs that describe the sender's strategy. The vertices of this graph are the states and equilibrium messages, with an edge connecting states with each message sent from that state with positive probability. Such a graph is a coarse representation of the equilibrium, as it omits the probability with which each message is sent, nevertheless its data about sender-indifferences are sufficient to study $\mathscr{O}$-stability.

Our analysis derives its structure from the two types of indifferences that are inherent to candidate equilibria. In the first place, the sender must be indifferent between messages to communicate. These indifferences generally cannot be obtained when the receiver only plays pure actions, so the receiver must be induced to randomize after most messages. This significantly constrains the equilibrium posterior beliefs of the receiver, and consequently the equilibrium strategies of the sender. In Section \ref{sec:receiver} we show that under weak (generic) conditions, this constrains the communication graph to be connected and/or contain a cycle.

Secondly, the indifferences given by the communication graph must be preserved in the approximating equilibrium. Under $\mathscr{O}$-stability, it is impossible for a cyclic set of indifferences to be robust to an open set of perturbations --- $\mathscr{O}$-stability requires an acyclic communication graph.

Combining these analyses shows that, under generic conditions, an equilibrium is $\mathscr{O}$-stable, for \emph{some} open set $\mathscr{O}$, iff its communication graph is acyclic and connected, ie. a tree graph.

This tree structure is essential for the next stage of our analysis which identifies precisely the maximal set of modifications $\mathscr{O}$ that makes an equilibrium $\mathscr{O}$-stable. This is precisely the set of perturbations satisfying a condition that we call \emph{graph monotonicity}.\footnote{Graph monotonicity bears some semblance to Mechanism Design's `cyclic monotonicity'. We omit the `cyclic' adjective to avoid confusion with the acyclicity of communication graphs.}

After characterizing $\mathscr{O}$-stability, which describes how communication can occur with slightly state-dependent preferences, Section \ref{sec:bigeps} studies whether this analysis can be extended to larger degrees of state-dependence. In general the answer is no, however we find a stronger condition than graph monotonicity that ensures an equilibrium can be preserved so long as the sender's ordinal ranking of pure actions remains state independent. This stronger condition is an analog to `single-crossing', weakened to apply to non-ordered settings.

We apply this technology to several example perturbations in Section \ref{sec:apps}. We first study ethical perturbations, relating to empathy and lying aversion. These describe situations where the sender's primary concern is their state-independent material payoff, however they defer to ethical considerations as a `tie-breaking' mechanism. Both empathy and lying aversion stabilize acyclic equilibria in two-state environments, but this does not always generalize to environments with more states.

The stabilizing influence of empathy is tied to whether the receiver and the candidate equilibrium order states in the same way. Specifically, if the receiver's problem satisfies a single crossing condition, providing an order on states and actions, then empathy stabilizes a candidate equilibrium if lower states recommend lower actions in that candidate equilibrium. A natural example of a salesperson selling vertically differentiated products illustrates how empathy may \emph{inhibit} communication when these conditions are not met. In this case persuasion requires the receiver to risk a costly error, but the empathetic sender is unwilling to induce this error, and only a pareto inferior babbling equilibrium survives.

Lying aversion, on the other hand, is in tension with the nature of candidate equilibria, which generally require obfuscating the state. Lies can be incentivized by a higher material payoff, however care must be taken that this does not interfere with the incentives of senders in other states. This is only possible for simple geometries of communication graphs --- specifically either there is only one state where the sender lies, or only one message is ever a lie.

It is possible to stabilize more candidate equilibria with modifications that are flexible with the truth. Our last example demonstrates this with a modification that allows the sender to make vague statements about the state, stabilizing all acyclic candidate equilibria. We interpret this example as communicating through a low-credibility disclosure technology --- it corresponds to perturbing the transparent preference model towards a model of verifiable disclosure.

We use the lying aversion model to explore how communicating through burning money may allow the sender to obtain higher utilities. This contrasts with candidate equilibria where burning money cannot improve equilibrium sender utility.
\subsubsection*{Related Literature}
This work primarily lies at the intersection of cheap talk communication models and Harsanyi purification literature. We provide a short genealogy of our model to situate our work:

\paragraph{Cheap Talk} The original cheap talk model (\cite{CS82}) specifically studies the case where the state-/action-space is $\Theta=\mathcal{A}=[0,1]$, the sender's preference is state-dependent, and the receiver has a unique preferred action for every belief. The existence of informative equilibria is a question of the parametric `alignment' between the sender and receiver preferences --- loosely resembling our notion of empathy. This requires state-dependence in the sender's preference: if the sender has a strict, state-independent preference (e.g. higher actions are better), then no information can be communicated in equilibrium --- if two messages led to different actions, the sender would always want to deviate to send the message corresponding to the preferred action. 

\cite{CH10} are the first (to my knowledge) to study the intriguing possibility of communication without state-dependence. They consider models with multi-dimensional state-/action-spaces $\Theta=\mathcal{A}\subseteq \mathbb{R}^n$ ($n>1$) where the receiver has quadratic preferences, and show that there exists an informative equilibrium for \emph{any} transparent sender preference. Their argument uses hyperplanes to partition the state space into sets whose induced actions the sender is indifferent between. This relies on the multi-dimensional space, as these partitions are obtained by continuously rotating the dividing hyperplane --- although a similar argument can establish persuasion in one dimensional state-spaces if the sender's preference exhibits quasiconvexity. The sender having a public preference is essential --- \cite{DK21} show that these equilibria vanish under slight uncertainty in the sender's preference, as Harsanyi's purification fails to apply (discussed more below). Our fragility theorem is a modest extension of this result in finite-action environments.

The model with a public sender preference is extended to a general setting by \cite{LR20}. They adopt a belief-based approach (as pioneered by \cite{KG11}), showing that communication can be understood in terms of the quasiconcave and quasiconvex envelopes of the sender's indirect utility over receiver beliefs. This shows that equilibria with communication exist iff there exists some Blackwell experiment that guarantees the sender a better/worse \emph{ex post} payoff than babbling. This result applies to models with finite state- and action-spaces, which is the environment we focus on.\footnote{We neglect the continuous state-space case in this paper. The topology of the belief space for continuous states adds complications that call for distinct techniques.} 

Independent of the current paper, \cite{SGGK23} consider a similar question of how to robustly approximate equilibria in binary state models. Specifically, they ask the question of which candidate equilibria can be approximated by Harsanyi-stable equilibria by introducing state dependence. We are able to study general finite-state models, due to a subtle difference in our notions of stability: their model adds state-dependence first, then adds some arbitrarily small idiosyncrasy --- in order of limits, this takes idiosyncrasy to zero, then state-dependence to zero --- while we entangle state-dependence and idiosyncrasy in the same perturbation $V$ (see eq. \ref{eq:intro}) --- thus requiring idiosyncrasy and state-dependence to vanish at the same rate. Our equilibria are thus robust to slight changes in state-dependence, which is not guaranteed for the former equilibria (for a more detailed comparison of the concepts, see Appendix \ref{app:weakstab}). That benefit aside, the key advantage of our refinement is that it allows tractable analysis of many-state models.


\paragraph{Purification} In addition to the aforementioned works, our work is also intimately related with the results of \cite{H73}, which studies mixed-strategy Nash equilibria in finite normal form games. The main result shows that (generically) such equilibria can be interpreted, not as players randomizing, but as players choosing pure strategies based on small idiosyncrasies in their preference --- so-called `purification' of mixed equilibria. This idiosyncrasy results in almost every player having a strict best response to the (seemingly) mixed strategies of other players. To other players, ignorant of the realized idiosyncrasies, the resulting actions appear mixed.

Harsanyi shows that adding idiosyncrasies is capable of approximating any mixed strategy equilibria in finite games for generic utility functions. But this result may fail for specific player preferences. A necessary condition is that if a mixed action $p$ is played in equilibrium, then each pure action $a\in\supp{p}$ should be a strict best response to some local variation in other players' actions (constrained to their equilibrium best responses).\footnote{Let $\mathcal{B}^i\subseteq \mathcal{A}^i$ be set of equilibrium best responses for player $i$. A more general requirement is that best response sets --- profiles of mixed actions $\text{BR}(A)\subseteq \bigtimes(\Delta\mathcal{B}^i)$ to which the subset of pure actions $A\subseteq\bigcup\mathcal{B}^i$ are best responses --- should not be degenerate, in that $\text{codim}(\text{BR}(A))= |A|-N(A)$ where $N(A)$ is the number of players whose actions feature in $A$.}

Transparent preferences may fail the preconditions to Harsanyi's theorem in two ways:
\begin{enumerate}
    \item In cheap talk, payoffs are invariant to message permutation, so if two messages induce the same pure action, then neither can be made a strict best response over the other.
    \item With state-independent preferences, a state-dependent strategy can never be a strict best response; its payoff can always be replicated by a state-independent strategy.
\end{enumerate}
In the first situation, a message-dependent perturbation may lead to either of these messages being sent much more (and perhaps from different states) than the other. However, this is a minor problem, as it can be avoided by eliminating duplicate messages in equilibrium. Indeed, whenever different messages lead to different mixed actions, it is possible to manipulate these actions to generate an arbitrary preference over messages in the sender population. By varying these actions, purification may be obtained.

The second problem occurs for any receiver strategy and hence is more intractable. It occurs because senders do not care about the state an action is chosen in, but only the probability an action is chosen. Thus, for \emph{any} receiver strategy, the sender is indifferent across strategies that induce the same marginal distribution over messages, regardless of how this distribution correlates with the state. Consequently, any non-constant (ie. informative) mapping from states to messages achieves the same utility as an appropriate randomization over constant (thus uninformative) mappings. As a result, there's no receiver strategy that can make a sender strictly prefer an informative strategy.

By adding even slight state dependence, we create the possibility for an informative strategy to be the unique best response to some set of receiver strategies, solving the second problem and allowing purification to potentially apply. However the persistence of the first problem means that Harsanyi's theorem cannot be applied out-of-the-box as some equilibria may remain fragile. This forces us to develop distinct techniques to show purification applies in our setting. 



\section{Model}
We consider a sender-receiver game, where the sender begins by observing the state, then chooses to send a message to the receiver. Upon reception, the receiver chooses an action, determining the utility obtained by both parties. 

The state space is formally a two dimensional probability space $(\Theta\times\Omega,2^\Theta\otimes\mathcal{F}_\omega,\mathbb{P}=\mu\otimes\mathbb{P}_\omega)$. The finite set $\Theta$ contains the receiver-relevant state distributed according to the prior $\mu$ --- we will use `state' to exclusively refer to these objects --- while $\Omega$ refers to sender `idiosyncrasy' types. These idiosyncrasies ar central to our notion of stability, as we seek equilibria that are robust to slight idiosyncratic variations in the sender's preference. Our framework assumes that states and idiosyncrasies are independent\footnote{This is WLOG, by redefining idiosyncrasies, we can reinterpret state-dependent idiosyncrasies as state-independent. See footnote \ref{fn:indep} for details.}. We assume that $(\Omega,\mathcal{F}_\omega,\mathbb{P}_\omega)$ is a rich probability space. 

The sender chooses a message from the message space $M$. We impose little structure on $M$, only equipping it with the discrete topology.

The receiver's action is chosen from the set of mixed actions $\Delta \mathcal{A}$ where the underlying pure action set $\mathcal{A}$ is finite. \footnote{In rare situations infinite actions may lead to pathological utility functions which require multiple technical assumptions to rule out. Our work can be applied to most `realistic' infinite action models with finite states.} 
We represent general actions in $\Delta \mathcal{A}$ by $p$ to emphasize their possibly mixed nature, reserving $a$ for actions that are restricted to be pure.

The sender and receiver have respective utility functions
\begin{align*}
    u_S:&(\Delta \mathcal{A}\times  M)\times (\Theta\times \Omega)\rightarrow \mathbb{R}\\
    u_R:&\Delta \mathcal{A}\times \Theta\rightarrow \mathbb{R}
\end{align*}
satisfying the von Neumann-Morgenstern axioms. We will say a sender preference is non-idiosyncratic if it is invariant over idiosyncrasies $\omega\in\Omega$, paralleling the notion of a state-independence preference being invariant over states $\theta\in\Theta$. We will find it convenient to adopt a compact designation for an action-message pair: $\pi\equiv (p,m)\in \Delta \mathcal{A}\times M$ (or $\alpha\equiv (a,m)\in \mathcal{A}\times M$ when the action is constrained to be pure). We sometimes refer to such $\pi$ as a message or an action to emphasize the relevant component.

We assume that the sender's utility only has a weak dependence on $\Theta\times \Omega$, admitting the following decomposition into a \emph{transparent} (ie. public, state-independent) component $u^0$, and a \emph{perturbation} $V$:
\begin{equation*}
    u_S(\pi\rvert \theta,\omega)=u^0(\pi)+\epsilon V(\pi\rvert \theta,\omega)
\end{equation*}
where $\epsilon$ controls how far the preference is from being transparent. Since state-dependence and idiosyncrasy factor through the perturbation $ V$, they are a second order effect, reflecting either slight adjustments to the sender's narrow material preferences, or slight ethical considerations.

Throughout this paper we will consider multiple varieties of perturbation, varying both whether it is idiosyncratic and state-dependent. In general we will use the letter `u' to represent state-independent components of the sender's preference and `v' for state-dependent components; capital letters will indicate a component is idiosyncratic, emphasizing its interpretation as a random variable. When a perturbation is not idiosyncratic or its idiosyncrasy $\omega$ is fixed, it is an element of $\mathbb{R}^{\mathcal{A}\times M\times\Theta}$ and we call it a \emph{modification} --- with perturbations it may be necessary to consider the population of idiosyncrasies, whereas modifications can be analyzed at the individual level.

While our motivation for this paper lies in studying cheap talk models, we allow the sender's utility to also feature message dependence. In this case, we are left with a signalling model with weak state dependence. Depending on where the message-dependence occurs, we may interpret the model in different ways. If the message-dependence is transparent, factoring through $u$, we interpret it as a money-burning model of communication\footnote{Money burning greatly increases the scope for communication, see Appendix \ref{app:burn} or \cite{ASB00}. This allows our results to be applied to settings such as the mediated communication environment of \cite{LY23}.}. If the message-dependence factors through the perturbation $V$, this may be interpreted as either a deontological moral cost to lying, or as a minor technical cost to the message space.

Our interest is in Perfect Bayesian Equilibria of this sender-receiver game. Formally, this is composed of a sender strategy $\mathscr{M}:\Theta\times \Omega \rightarrow \Delta M$, a receiver posterior posterior belief function $\nu:M\rightarrow \Delta\Theta$\footnote{For simplicity, we omit the belief over $\Omega$ as sender idiosyncrasy is irrelevant to the receiver}, and a receiver action $\mathscr{P}:M\rightarrow \Delta\mathcal{A}$ such that
\begin{subequations}
\begin{align}\label{eq:sIC}
    \mathscr{M}(\theta,\omega)\in&\argmax_{m\in \Delta M}u_S(\mathscr{P}(m),m\rvert \theta,\omega)\\
    \nu(\theta\rvert m)=&\Pcond{\theta}{ \mathscr{M}(\theta,\omega)=m}\qquad \text{if }m\in\supp{\mathscr{M}}
    \\\label{eq:rIC}
    \mathscr{P}(m)\in&\argmax_{p\in \Delta \mathcal{A}}\int u_R(p\rvert \theta)\,d\nu(\theta).
\end{align}
\end{subequations}
We denote the set of these equilibria by $\Sigma(u_S,u_R)$ if $u_R$ is allowed to vary (in Section \ref{sec:gen} we do this to find the generic structure), and $\Sigma(u_S)$ when we fix $u_R$.

As $\epsilon\rightarrow 0$, these equilibria will converge to equilibria of the transparent model $\sigma\in\Sigma(u^0,u_R)$, which we refer to as \emph{candidate equilibria}. We use this term as by analyzing these (fragile) candidate equilibria, we can determine if a nearby robust equilibrium exists when the sender has preference $u_S$. 

We are particularly interested in `non-trivial' equilibria, where the sender successfully persuades the receiver with positive probability:
\begin{defn}
    An equilibrium $\sigma$ is \emph{persuasive} if there is positive probability that the receiver chooses an action that is not a best response to their prior belief:
    $$
    \P{\mathscr{P}(\mathscr{M}(\theta,\omega))\in \argmax_{p\in\Delta\mathcal{A}}\int u_R(p\rvert \theta)\,d\mu(\theta)}<1.
    $$
\end{defn}

We maintain the following assumption throughout the paper:
\begin{axiom}[Assumption (S)]
For any pair of distinct pure actions $a\neq a'$, we have $u^0(a,m)\neq u^0(a',m')$ for any messages $m,m'\in M$.
\end{axiom}
Within cheap talk this is merely the assumption that $u$ is injective on the set of pure actions $\mathcal{A}$.

The content of this assumption is that the sender is never indifferent between messages that lead to different pure actions --- consequently mixed actions are necessary to obtain persuasive equilibria.

\subsection{Candidate Equilibria}\label{sec:candid}
We follow the belief-based approach of \cite{LR20} to recast the problem of finding perfect bayesian candidate equilibria. In this case of transparent sender preferences $u^0$, eq. \ref{eq:sIC} becomes 
\begin{equation*}
    \mathscr{M}(\theta)\in\argmax_{m\in \Delta M}u^0(\mathscr{P}(m),m)
\end{equation*}
and thus all messages sent by the sender's strategy must result in the same expected utility $\overline{u}$. We can then characterize equilibria by this set of messages $M_0$, a mapping $\nu:M_0\rightarrow \Delta \Theta$ from messages to the induced posterior beliefs, and action strategies $\mathscr{P}:\nu(M_0)\rightarrow  \Delta \mathcal{A}$.

To satisfy receiver-incentive compatibility (eq. \ref{eq:rIC}), given a belief $\nu\in\mathcal{B}:=\nu(M_0)$ the receiver must choose a best response $p$. This induces an indirect utility correspondence $u^\ast:\Delta \Theta\times M\Rightarrow \mathbb{R}$ for the sender over the receiver's belief:
\begin{equation}
    u^\ast(\nu,m):=\left\{u^0(p,m);p \in\argmax_{p'\in\Delta \mathcal{A}} \int u_R(p'\rvert \theta)\,d\nu(\theta)\right\}.
\end{equation}
Note that this will indeed be multi-valued for some belief $\nu$, unless the receiver possesses a dominant action (in which case the model is trivial).

Sender incentive compatibility (eq. \ref{eq:sIC}) requires that senders must be indifferent between messages they do send. This can be done in a way that delivers utility $\overline{u}$ to the sender whenever
\begin{subequations}
\begin{equation}\label{eq:IC}
    \overline{u}\in\bigcap_{m\in M_0}u^\ast(\nu(m),m).
\end{equation}
Where $\overline{u}$ is bounded by off-path temptations:
\begin{equation}\label{eq:offpath}
    \overline{u}\ge \max_{m\in M}\min_{\nu\in\Delta \Theta}\{u^\ast(\nu,m)\}.
\end{equation}

Meanwhile, the receiver's Bayesianism impose a \emph{plausibility constraint} on the induced beliefs given the prior belief $\mu$:
\begin{equation}\label{plaus}
    \mu\in\text{co}(\mathcal{B}),
\end{equation}
where $\text{co}(\cdot)$ is the convex hull operation on sets.

\end{subequations}

Thus an equilibrium consists of splitting the prior belief $\mu$ into a set of posterior beliefs $\mathcal{B}$ that all attain the same sender-utility. In Appendix \ref{app:burn} we illustrate how the notion of candidate equilibria is extremely powerful in money-burning models, allowing persuasion generically in non-trivial models. 

In the cheap talk case, the sender's preference is message-independent, so eq. \ref{eq:offpath} is trivial. Thus, as \cite{LR20} observe, the solution can be characterized by a set of beliefs $\mathcal{B}\subseteq \Delta \Theta$ and a sender utility $\overline{u}$ satisfying
\begin{align*}
    \mu\in&\text{co}(\mathcal{B})&\text{and}&&
    \overline{u}\in& \bigcap_{\nu\in\mathcal{B}}u^\ast(\nu).
\end{align*}
The following example from their paper illustrates how their model provides persuasive (and pareto improving) equilibrium whenever the sender finds the default action of the receiver (ie. their action under the prior belief) \emph{ex post} inferior to full revelation\footnote{More generally, persuasive cheap talk equilibria exist whenever $|M|\ge|\Theta|$ and the receiver's best response to $\mu$ is \emph{ex post} inferior to the receiver's actions after some Blackwell experiment.}:


\begin{ex}[Binary-state salesperson] \label{ex:bi}
Suppose there is a salesperson trying to convince a consumer to make a purchase between two products $A$ and $B$. There are two possible states of the world: in state $\theta=\theta_a$ product $A$ is good and product $B$ is bad, while in state $\theta=\theta_b$ product $B$ is good while product $A$ is bad. The consumer believes each state is equally likely: their prior belief is $\mu=\frac{1}{2}(s_a\oplus s_b)$. For simplicity, we refer to a belief $\nu\in\Delta\{\theta_a,\theta_b\}$ by the probability $\nu(\theta_b)$ it puts on state $\theta_b$.

The salesperson knows which product is of good quality, but only cares about their commision, which is higher for product $B$ than for product $A$. Based on the salesperson's recommendation, the consumer can purchase either product, or walk away without making a purchase --- we label their action by the product purchased $\mathcal{A}=\{A,B,\emptyset\}$, where $\emptyset$ means no purchase. This leads to the following utility tables for the receiver (consumer) and sender (salesperson) respectively:
\begin{center}
    \begin{tabular}{r|c|c|c}
         $_\theta\smallsetminus ^\alpha$&$A$&$\emptyset$&$B$  \\\hline
         $u_R(\,\cdot\,\rvert \theta_a)$& 1&3/4&0\\
         $u_R(\,\cdot\,\rvert \theta_b)$&0&3/4&1
    \end{tabular}\hspace{2cm}
    \begin{tabular}{r|c|c|c}
         &$A$&$\emptyset$&$B$  \\\hline
         $u^0$& 1/2&0&1
    \end{tabular}
\end{center}

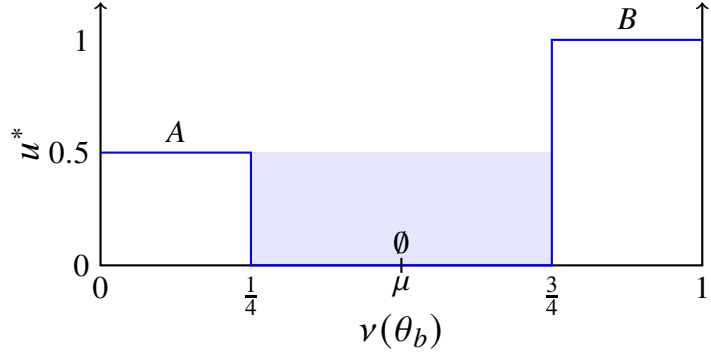
\begin{figure}
    \centering
    \begin{tikzpicture}
    \fill[blue!10!white] (2,1.5) -- (6,1.5) -- (6,0) -- (2,0) -- cycle;
    \draw[thick,<->] (0,3.5) -- (0,0) -- (8,0) -- (8,3.5);
    \draw[thick, blue] (0,1.5) -- (2,1.5) -- (2,0) -- (6,0) -- (6,3) -- (8,3);
    \draw[thick] (4,0.1) -- (4,-0.1);
    \node[anchor = east] at (0,0) {$0$};
    \node[anchor = east] at (0,1.5) {$0.5$};
    \node[anchor = east,rotate=90] at (-1,2) {\large$u^\ast$};
    \node[anchor = east] at (0,3) {$1$};
    \node[anchor = north] at (4,0) {$\mu$};
    \node[anchor = north] at (4,-0.5) {{\large$\nu(\theta_b)$}};
    \node[anchor = south] at (7,3) {$B$};
    \node[anchor = south] at (1,1.5) {$A$};
    \node[anchor = south] at (4,0) {$\emptyset$};
    \node[anchor = north] at (0,0) {$0$};
    \node[anchor = north] at (2,0) {$\tfrac{1}{4}$};
    \node[anchor = north] at (6,0) {$\tfrac{3}{4}$};
    \node[anchor = north] at (8,0) {$1$};
    \end{tikzpicture}
    \caption{The sender's indirect utility in Example \ref{ex:bi}. The shaded region is the range of utilities attainable by cheap talk for a given prior.}
    \label{fig:LRex}
\end{figure}

The sender's indirect utility for this problem is illustrated in Fig. \ref{fig:LRex}. Notably, full-revelation \emph{ex post} dominates babbling. The sender-utilities that are attainable for a given prior are given by the blue shaded region above the prior. The sender can obtain utility $\overline{u}$ only if $\overline{u}\in[0,1/2]$, since these are the only utilities supported on both sides of the prior belief.

This can be done by inducing the beliefs $\nu_0=1/4$ and $\nu_1=3/4$, convincing the receiver to choose $\mathscr{P}(\nu_0)=2\overline{u}A\oplus(1-2\overline{u})\emptyset$ and $\mathscr{P}(\nu_1)=\overline{u}B\oplus(1-\overline{u})\emptyset$ respectively. These posteriors correspond to the sender playing the mixed strategy $\mathscr{M}(\theta_a)=\frac{3}{4}m_0\oplus \frac{1}{4}m_1$ and $\mathscr{M}(\theta_b)=\frac{1}{4}m_0\oplus \frac{3}{4}m_1$.

The utility $\overline{u}=1/2$ can also be obtained by inducing posteriors $\nu_0\le 1/4$ and $\nu_1=3/4$, convincing the receiver to play $\mathscr{P}(\nu_0)=A$ and $\mathscr{P}(\nu_1)=\frac{1}{2}(B\oplus\emptyset)$ respectively. This corresponds to the sender playing the mixed strategy $\mathscr{M}(\theta_a)=\frac{2(1-\nu_0)}{3-4\nu_0} m_0\oplus\frac{1-2\nu_0}{3-4\nu_0} m_1$ and $\mathscr{M}(\theta_b)=\frac{2\nu_0}{3-4\nu_0}m_0\oplus\frac{3-6\nu_0}{3-4\nu_0} m_1$. The case $\nu_0=0$ deserves particular attention: not only does it correspond to partial revelation, but it is in some sense the `simplest', and furthermore both receiver- and sender-optimal.
\end{ex}
Note that there are a continuum of equilibria in this example. Under assumption (S), such models either induce a continuum of equilibria or no persuasive equilibrium --- as the intersection in eq. \ref{eq:IC} is either always empty or can be made to contain an interval. Over the next few sections we will establish a robustness criteria that eliminates all but a finite number of these candidate equilibria.

\section{Communication Graphs}\label{sec:receiver}\label{sec:gen}
Our analysis characterizes equilibria by the on-path subgames they induce. Formally, given an equilibrium $\sigma$ composed of strategies $(\mathscr{M}_\sigma,\mathscr{P}_\sigma)$, define $\sigma(\theta):=\{(p,m);m\in\supp{\mathscr{M}_\sigma(\theta)},p\in\mathscr{P}_\sigma(m)\}$ to be the equilibrium paths induced by the event that the state $s$ is realized --- ie. message-action pairs that occur when the state is $\theta$.\footnote{For cheap talk games (ie. when preferences are constrained to be message independent), one should consolidate messages that result in the same choice of (possibly mixed) action.

Note that this structure omits the probability with which a subgame $\sigma(\theta)$ occurs. We will find that this probability is identified by a communication graph for generic acyclic equilibria.}

\begin{defn}
    The \emph{communication graph} $G(\sigma)$ associated with an equilibrium $\sigma$ is the undirected bipartite graph whose nodes are $\Theta\sqcup\sigma(\Theta)$ and whose edges are
    $\big\{\{\theta,\pi\};\pi\in \sigma(\theta),\theta\in \Theta\big\}.$
    
    A communication graph (and the equilibrium it represents) is 
\begin{itemize}
    \item \emph{connected} if it contains a path between any two nodes,
    \item \emph{acyclic} if it contains at most one path between any two nodes,
    \item a \emph{tree} if it is both acyclic and connected,
    \item a \emph{forest} if it is acyclic but not connected.
\end{itemize}
\end{defn}

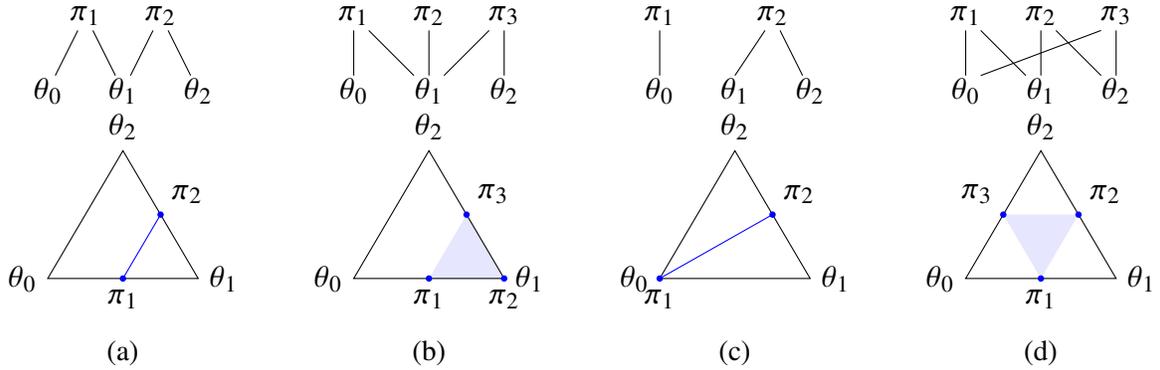
\begin{figure}
    \centering
    \begin{subfigure}{0.24\textwidth}\centering
    \begin{tikzpicture}
        \node at (3,0) {$\theta_2$};
        \node at (2,0) {$\theta_1$};
    \node at (1,0) {$\theta_0$};
    \node at (1.5,1) {$\pi_1$};
    \node at (2.5,1) {$\pi_2$};
    \draw (1.1,0.2) -- (1.4,0.8);
    \draw (2.1,0.2) -- (2.4,0.8);
    \draw (1.9,0.2) -- (1.6,0.8);
    \draw (2.9,0.2) -- (2.6,0.8);
    \draw (1,-2.5) -- (2,-0.8) -- (3,-2.5) -- cycle;
    \node[anchor = east] at (1,-2.5) {$\theta_0$};
    \node[anchor = south] at (2,-0.8) {$\theta_2$};
    \node[anchor = west] at (3,-2.5) {$\theta_1$};
    \draw[blue] (2,-2.5) -- (2.5,-1.65);
    \filldraw [blue] (2,-2.5) circle (1 pt) node[anchor = north] {\color{black} $\pi_1$};
    \filldraw [blue] (2.5,-1.65) circle (1 pt) node[anchor = south west] {\color{black} $\pi_2$};
    \end{tikzpicture}
    \caption{}
    \label{fig:graphtreea}
    \end{subfigure}
    \begin{subfigure}{0.24\textwidth}\centering
    \begin{tikzpicture}
        \node at (7,0) {$\theta_2$};
        \node at (6,0) {$\theta_1$};
    \node at (5,0) {$\theta_0$};
    \node at (5,1) {$\pi_1$};
    \node at (6,1) {$\pi_2$};
    \node at (7,1) {$\pi_3$};
    \draw (5,0.2) -- (5,0.8);
    \draw (6,0.2) -- (6,0.8);
    \draw (7,0.2) -- (7,0.8);
    \draw (6.2,0.2) -- (6.8,0.8);
    \draw (5.8,0.2) -- (5.2,0.8);
    \fill[blue!10!white] (6,-2.5) -- (6.5,-1.65) -- (7,-2.5) -- cycle;
    \draw (5,-2.5) -- (6,-0.8) -- (7,-2.5) -- cycle;
    \node[anchor = east] at (5,-2.5) {$\theta_0$};
    \node[anchor = south] at (6,-0.8) {$\theta_2$};
    \node[anchor = west] at (7,-2.5) {$\theta_1$};
    \filldraw [blue] (6,-2.5) circle (1 pt) node[anchor = north] {\color{black} $\pi_1$};
    \filldraw [blue] (6.5,-1.65) circle (1 pt) node[anchor = south west]  {\color{black} $\pi_3$};
    \filldraw [blue] (7,-2.5) circle (1 pt) node[anchor = north] {\color{black} $\pi_2$};
    \end{tikzpicture}
    \caption{}
    \label{fig:graphtreeb}
    \end{subfigure}
    \begin{subfigure}{0.24\textwidth}\centering
    \begin{tikzpicture}
    \node at (11,0) {$\theta_2$};
    \node at (10,0) {$\theta_1$};
    \node at (9,0) {$\theta_0$};
    \node at (9,1) {$\pi_1$};
    \node at (10.5,1) {$\pi_2$};
    \draw (9,0.2) -- (9,0.8);
    \draw (10,0.2) -- (10.4,0.8);
    \draw (10.9,0.2) -- (10.6,0.8);
    \draw[blue] (9,-2.5) -- (10.5,-1.65);
    \draw (9,-2.5) -- (10,-0.8) -- (11,-2.5) -- cycle;
    \node[anchor = east] at (9,-2.5) {$\theta_0$};
    \node[anchor = south] at (10,-0.8) {$\theta_2$};
    \node[anchor = west] at (11,-2.5) {$\theta_1$};
    \filldraw [blue] (9,-2.5) circle (1 pt) node[anchor = north] {\color{black} $\pi_1$};
    \filldraw [blue] (10.5,-1.65) circle (1 pt) node[anchor = south west] {\color{black} $\pi_2$};
    \end{tikzpicture}
    \caption{}
    \label{fig:graphforest}
    \end{subfigure}
    \begin{subfigure}{0.24\textwidth}\centering
    \begin{tikzpicture}
    \node at (15,0) {$\theta_2$};
    \node at (14,0) {$\theta_1$};
    \node at (13,0) {$\theta_0$};
    \node at (13,1) {$\pi_1$};
    \node at (14,1) {$\pi_2$};
    \node at (15,1) {$\pi_3$};
    \draw (14,0.2) -- (14,0.8);
    \draw (15,0.2) -- (15,0.8);
    \draw (13,0.2) -- (13,0.8);
    \draw (13.2,0.8) -- (13.8,0.2);
    \draw (14.2,0.8) -- (14.8,0.2);
    \draw (13.2,0.2) -- (14.8,0.8);
    \draw (13,-2.5) -- (14,-0.8) -- (15,-2.5) -- cycle;
    \node[anchor = east] at (13,-2.5) {$\theta_0$};
    \node[anchor = south] at (14,-0.8) {$\theta_2$};
    \node[anchor = west] at (15,-2.5) {$\theta_1$};
    \fill[blue!10!white] (14,-2.5) -- (13.5,-1.65) -- (14.5,-1.65) -- cycle;
    \filldraw [blue] (14,-2.5) circle (1 pt) node[anchor = north] {\color{black} $\pi_1$};
    \filldraw [blue] (14.5,-1.65) circle (1 pt) node[anchor = south west] {\color{black} $\pi_2$};
    \filldraw [blue] (13.5,-1.65) circle (1 pt) node[anchor = south east]  {\color{black} $\pi_3$};
    \end{tikzpicture}
    \caption{}
    \label{fig:graphcyc}
    \end{subfigure}
    \caption{Example communication graphs (top) and corresponding induced beliefs (bottom) in a three-state game. The nodes of the graphs correspond to states $\{\theta_i\}$ and message-action pairs $\{\pi_i\}$ corresponding to the beliefs illustrated below. The line/shaded regions within the simplices correspond to the priors that render the given posteriors plausible --- ie. the convex hull of these posteriors. (a), (b) are trees, (c) is a forest; only (d) contains a cycle.}
    \label{fig:graphs}
\end{figure}
Note that communication graphs are a coarse representation of equilibria, as they do not specify the specific posterior beliefs induced, merely what states the posterior belief designates as \emph{plausible} (ie. within its support). However this, along with the receiver's strategy $\mathscr{P}$, is sufficient to study sender-incentive compatibility in our setting.

Some example communication graphs, and the corresponding beliefs, are illustrated in Figure \ref{fig:graphs}.

Acyclic equilibria can loosely be thought as a class of `simple' equilibria, not requiring the receiver's mixed actions to maintain a complex network of sender indifferences. For our purposes, acyclicity implies that the message structure can be partially ordered relative to some `root' node, allowing an inductive analysis of equilibria.

Non-connected equilibria can be thought of as a class of `separating' equilibria. For such equilibria there is a partition $\Theta=\Theta_1\sqcup \Theta_2$ that is communicated with certainty, ie. for all posteriors $\nu\in\mathcal{B}_\sigma$, either $\supp{\nu}\subseteq \Theta_1$ or $\supp{\nu}\subseteq \Theta_2$. In this sense non-connectedness is related to the notion of a separating equilibria in signalling games --- in contrast with the (partial) pooling of connected equilibria.

\label{sec:gentop}

Taking these properties to the extreme, we arrive at pure strategy equilibria, where $|\sigma(\theta)|=1$ for all states.\footnote{See \cite{GS07} for a notable analysis of pure equilibria in finite-state/-action models.} These are acyclic (indeed there is no path between any two messages), and, aside from babbling, non-connected. Due to the requirement, imposed by Assumption (S), that receiver's randomize after all but one message, such equilibria are rare:
\begin{fact}
Persuasive pure-strategy candidate equilibria exist only if there is a partition of states $\Theta=\Theta_0\sqcup \dots \sqcup \Theta_n$ such that the conditioned prior beliefs $$
\mu_{\Theta_i}:=\left(\tfrac{\mu(\{\theta\}\cap \Theta_i)}{\mu(\Theta_i)}\right)_{\theta\in \Theta}
$$ induce mixed actions for $i\ge1$. 
\end{fact}
For example, Figure \ref{fig:graphforest} can only be an equilibrium if the receiver is indifferent between two actions at the precise belief induced by $\pi_2$, corresponding to a projection of the prior $\mu$ onto $\Delta\{\theta_1,\theta_2\}$. This requires very specific receiver utilities; whenever this indifference occurs, a slight perturbation of receiver utilities will make such an equilibrium impossible. 

This may be attributed to the set of pure strategies only being capable of inducing a finite set of possible posterior distributions $\mathcal{B} \subset \Delta \Theta$ which are unlikely to coincide with the (usually) measure-0 set of posteriors inducing randomization. This observation is generalized in the following theorem:

\begin{thm}[The generic case]\label{thm:gen}
For generic $u_R\in\mathbb{R}^{\mathcal{A}\times \Theta}$, the set of candidate equilibria $\Sigma(u^0,u_R)$ satisfies the following properties:
\begin{enumerate}[label=(\textbf{\alph*})]
    \item there are no forest equilibria,
    \item all acyclic equilibria involve precisely one pure action.
    \item all mixed actions in acyclic equilibria are binary (ie. $p$ with $|\supp{p}|=2$),
    \item there are finitely many acyclic equilibria (up to information equivalence\footnote{Two equilibria are informationally equivalent if they induce the same distribution of posterior beliefs in the receiver.}), each uniquely identified by its communication graph.
\end{enumerate}
\end{thm}
The focus on acyclic equilibria will be justified in the next section when we focus on the sender's problem. 

The reasoning is in essence a dimensionality argument. Each sender indifference allows another degree of freedom (randomization) in the sender's strategy, thus a tree equilibrium with on-path messages $M_0$ allows posteriors to vary over an $|M_0|-1$-dimensional subset of $(\Delta \Theta)^{M_0}$. A fixed receiver's strategy is also a best response to some subset of $(\Delta \Theta)^{M_0}$. For the receiver's strategy to be receiver-incentive compatible with a sender strategy having this communication graph, these sets must intersect, (generically) requiring a receiver strategy that is a best response to a set of codimension at most $|M_0|-1$. However, if every action taken is a binary mixed action, the corresponding strategy is a best response to a set of beliefs with codimension $|M_0|$ (the intersection of $|M_0|$ sets of codimension 1) for generic $u_R$. The only way to reduce this codimension is by replacing one mixed action with a pure action, giving properties (b)-(c).

For each tree communication graph, these intersections occur finitely many times (each corresponding to a distinct set of mixed actions), and since there are finitely many trees with states $\Theta$ (up to permutation of messages), this results in finitely many acyclic equilibria, property (d).

Forest equilibria are composed of at least two trees, each requiring a pure action (by the above reasoning), since the sender is never indifferent between distinct pure actions, this is not sender-incentive compatible, so forest equilibria are generically impossible, property (a).

In Appendix \ref{app:proofs}, we show that properties (a)-(d) are also obtained for generic prior $\mu$, under mild assumptions on $u_R$. 

This elimination of forest equilibria will allow us to simplify our analysis in the next section.

\section{Fragility}\label{sec:frag}
Our first robustness question is how the slightest amount of sender idiosyncrasy affects persuasive equilibria. This is desirable since perfect knowledge of another's preferences (including, necessarily, their risk aversion) is implausible.

For this section it suffices to consider state-independent perturbations to the senders utility: $V\equiv U:\Delta\mathcal{A}\times M\times\Omega\rightarrow \mathbb{R}$. We begin with a definition from \cite{H73}. Take the notation $\Gamma((u_i)_i)$ to refer to a game $\Gamma$ where player $i$ has utility function $u_i$. For idiosyncratic perturbations $(U_i)$, consider the perturbed utility functions
$$
u_i^{\epsilon U_i}(\pi\rvert \omega):=u^0_i(\pi)+\epsilon U_i(\pi\rvert \omega).
$$
We understand a candidate equilibrium as approximating an environment with idiosyncrasy when the following condition is satisfied:
\begin{defn}[Harsanyi stability]
An equilibrium $\sigma\in \Sigma(u^0_S,u_R^0)$ of the game with preferences $u_i^0$ is \emph{Harsanyi-stable} if for all bounded perturbations $U_i$, there exist equilibria $\sigma^\epsilon\in\Sigma(u_S^{\epsilon U_S},u_R^{\epsilon U_R})$ of the game with preferences $u_i^{\epsilon U_i}$ such that $\sigma^\epsilon\rightarrow \sigma$ in distribution as $\epsilon\rightarrow 0$.

Otherwise the equilibrium is \emph{Harsanyi-fragile}.

\end{defn}

We will focus on sender idiosyncrasies, as receiver idiosyncrasies are generically unproblematic.

\cite{H73} shows that this notion of stability is a property of Nash Equilibria in generic finite games. While these idiosyncrasies change individual players' preferences, other players, ignorant to the realization of these idiosyncrasies, will not observe a difference in probabilistic behaviour. This provides an interpretation of mixed strategy equilibria as being a manifestation of unobserved idiosyncrasies.


Harsanyi's argument relies on the strategy space not featuring redundancies. This condition fails in models of cheap talk with transparent preferences for two distinct reasons:
\begin{enumerate}
    \item The sender's utility is invariant to relabelling of messages.
    \item The sender's utility is invariant to `informativeness'. Fixing the receiver with an arbitrary strategy $\mathscr{P}$, the sender's utility is equal under the following strategies: 
    \begin{enumerate}
        \item Informative strategy: send each message $m$ with a state-dependent probability $p_m(\theta)$.
        \item Uninformative strategy: send each message $m$ with a state-independent probability $\sum_\theta p_m(\theta)\mu(\theta)$ equal to the total probability that the message is sent under the informative strategy.
    \end{enumerate}
\end{enumerate}
The fragility associated with the first redundancy is easily resolved --- either by selecting equilibria without redundant messages, or by forcing the model to remain in the realm of cheap talk with message-independent idiosyncrasies.

However the second type of redundancy still occurs, and turns out to be more intractable:

\begin{prop}[Adaptation of \cite{DK21}\footnote{The referenced result establishes the fragility of persuasive equilibria in the $\Theta=\mathcal{A}=\mathbb{R}^n$ model of \cite{CH10}. (1) adapts their techniques to finite pure action models (where their Condition (S) may not hold), leading to the novel result (2).}]\label{prop:Frag}
All persuasive equilibria of communication games where the sender has state-independent preferences are Harsanyi fragile:
\begin{enumerate}
    \item No persuasive equilibrium exists if the sender's preference may be decomposed 
    \begin{equation}\label{eq:frag}
        u_S(a,m\rvert (\omega_1,\omega_2)):=\tilde{u}(a,m\rvert \omega_1)+ U(a\rvert \omega_2)
    \end{equation}
    where $U$ is a message-independent component that admits a density over $\mathbb{R}^\mathcal{A}$.\footnote{If the reader wishes to normalize preferences, consider $U$ that admit a density over non-normalized actions.}
    \item For fixed receiver utility $u_R$, persuasion is (topologically) generically impossible over the set of preferences $\Delta \mathbb{R}^{\mathcal{A}\times M}$ under the weak-$\ast$ topology.
\end{enumerate}
\end{prop}
These results rely only on the state-independence of $u_S$ and the finiteness of $\mathcal{A}$, and can be extended to settings with infinite state spaces and idiosyncratic receivers (ie. where the receiver's preference depends on an idiosyncrasy $\omega_R$ which is distributed independently of the sender's idiosyncrasy $\omega$).

Equation \ref{eq:frag} represents preferences that have even mild (and potentially hidden) correlation in the value of actions across different messages. Consequently, cheap talk persuasion is impossible whenever the sender's preference is distributed according to a state-independent density --- no matter how concentrated --- and, even if persuasion is possible for some rare utility $\tilde{u}$, this persuasion can be destroyed by adding even a small idiosyncratic perturbation that is independent of the first.

The intuition behind this result is that a sender who first observes their idiosyncrasy $\omega$ will, w.p. 1, have a strict preference over equilibrium messages, and thus will choose a message independent of their observation of $s$.

In this paper we are primarily interested in finite message equilbiria: 


\begin{proof}[Proof of (1) when $M$ is finite.]
Suppose there exists a persuasive equilibrium $\sigma$, and let $p_i$ be a mixed action that is chosen with positive probability. Let $M_i$ be the set of messages that induce an action $p_i$, and $\Phi_i:=\{p_i\}\times M_i$.

The probability that a message in $M_i$ is sent from state $s$ is then bounded between two state-independent quantities:
\begin{equation*}
\begin{split}
    \mathbb{P}\Big[\max_{\pi_i\in \Phi_i}u_S(\pi_i)> \max_{\pi_j\not\in \Phi_i} u_S(\pi_j)\Big] &\le\Pcond{M_i}{\theta}\le\\ \mathbb{P}\Big[\max_{\pi_i\in \Phi_i}u_S(\pi_i)> \max_{\pi_j\not\in \Phi_i}u_S(\pi_j)\Big]&+\mathbb{P}\Big[\max_{\pi_i\in \Phi_i}u_S(\pi_i)=\max_{\pi_j\not\in \Phi_i}u_S(\pi_j)\Big],
\end{split}
\end{equation*}
where the last term is dominated by a finite sum of probability 0 sets --- one for each other message. Thus every action is recommended with a positive state-independent probability, and hence is optimal under the prior belief.\footnote{This proof does not extend to cases where a continuum of messages are sent. In this case, each sender $\omega$ may be indifferent between messages $\pi_\omega,\pi_\omega'$ varying with $\omega$, and randomizes between these messages in an informative way. Such messages are each sent with probability zero, but cumulatively are sent with positive probability. This general case is treated in Appendix \ref{app:proofs}.}
\end{proof}


The result of this section is that persuasion is impossible with state-independent sender preferences if there is \emph{any} fuzziness about what these preferences are. Since perfect knowledge about other agents is impossible in reality, this suggests that practical persuasion relies on a state-dependent element entering the model.

\section{$\mathscr{O}$-Stability and the Sender's Problem}\label{sec:stab}
In this section, our main question is whether adding slight state-dependence to the sender's preference can re-establish persuasive equilibria. We start by studying non-idiosyncratic perturbations to simplify our arguments (our perturbation is $V\equiv v:\Delta\mathcal{A}\times M\times \Theta\rightarrow \mathbb{R}$) before re-introducing idiosyncrasy to establish robustness of our approximating equilibria.

Suppose the transparent sender preference $u^0$ is perturbed by the modification $v\in\mathbb{R}^{(\mathcal{A}\times M) \times \Theta}$:
$$u_S^{\epsilon v}(\pi\rvert \theta)=u^0(\pi)+\epsilon v(\pi\rvert \theta).$$
To obtain Harsanyi stability, an equilibrium must be robust to any modification. Since Proposition \ref{prop:Frag} shows Harsanyi stability is impossible, we focus on robustness to \emph{certain} modifications:
\begin{defn}
    For a nonempty open set of modifications $\mathscr{O}\subseteq \mathbb{R}^{(\mathcal{A}\times M) \times \Theta}$, we say a candidate equilibrium $\sigma\in\Sigma(u^0)$ is $\mathscr{O}$\emph{-stable} if for any $v\in\mathscr{O}$ there are equilibria $\sigma^{\epsilon}\in\Sigma(u_S^{\epsilon v})$ of the game with sender preferences $u^{\epsilon v}_S$ such that $\sigma^{\epsilon}\rightarrow \sigma$ as $\epsilon \rightarrow 0$.
\end{defn}
The notion of $\mathscr{O}$-stability parallels Harsanyi stability, however it weakens the condition in two distinct ways: (1) we consider only non-idiosyncratic perturbations, and (2) we allow ourselves to constrain these perturbations to an open set --- openness ensuring we avoid knife-edge scenarios. In Section \ref{sec:Harstab} we consider a version of $\mathscr{O}$-stability that allows for idiosyncrasy and obtain parallel results, showing relaxation (1) is merely simplifying while (2) is essential.

We will see that many candidate equilibria fail to be $\mathscr{O}$-stable, and indeed no persuasive candidate equilibrium is $\mathbb{R}^{(\mathcal{A}\times M) \times \Theta}$-stable.
To establish intuition for how a modification affects an equilibrium, consider the following first order approximation:
\begin{equation}
    u_S^{\epsilon v}(p+\Delta p,m\rvert \theta)=u^0(p,m)+\epsilon v(p,m\rvert \theta)+\langle (u^0_a),\Delta p\rangle+ O(\epsilon\Delta p)
\end{equation}
where $(u^0_a)$ is the vector of pure action utilities.\footnote{Readers may recognize an analogy to the standard mechanism design model through the following correspondence:
\begin{equation*}
    \begin{pmatrix}
        \sigma(\theta)\\
        \text{$\langle (u^0_a),\Delta p\rangle$}
    \end{pmatrix}\Longleftrightarrow
    \begin{pmatrix}
        \text{decision rule}\\
        \text{transfer mechanism}
    \end{pmatrix}.
\end{equation*}
There are two distinctions: (1) we must maintain some sender indifferences, (2) our stability criteria.
} Recalling our shorthand $\pi\equiv(p,m)$, note that for a candidate equilibrium $\sigma\in \Sigma(u^0)$, $u^0(\pi)$ is constant across $\pi\in \sigma(\Theta)$; the preference over equilibrium paths is determined by higher order terms.

Suppose $(\pi_1,\theta_2,\pi_2)$ is a path on the communication graph $G(\sigma)$ To maintain the $\theta_1$ sender's indifference, we must have
\begin{equation*}
    v(\pi_1\rvert \theta_2)-v(\pi_2\rvert \theta_2)=\left\langle (u^0_a),\tfrac{\Delta p_2-\Delta p_1}{\epsilon}\right\rangle+ O(\Delta p)
\end{equation*}
where $\Delta p_i$ is a perturbation to the candidate equilibrium mixed action specified by $\pi_i$. If moreover $(\pi_1,\theta_2,\dots,\theta_{N},\pi_N)$ is a path on $G(\sigma)$, we can sum the corresponding constraints to get
\begin{equation}\label{eq:neighs}
    \sum_{i=2}^{N}v(\pi_{i-1}\rvert \theta_i)-v(\pi_{i}\rvert \theta_i)=\left\langle (u^0_a),\tfrac{\Delta p_{N}-\Delta p_1}{\epsilon}\right\rangle+ O(\Delta p).
\end{equation}
But the righthand side also appears in the difference of utilities between $\pi_{1}$ and $\pi_{N}$. Consequently, if $\pi_1\in \sigma(\theta_1)$, we must have
\begin{equation}\begin{split}\label{eq:Gintuit}
    u_S^{\epsilon v}(p_1+\Delta p_1,m_1\rvert \theta_1)-u_S^{\epsilon v}(p_{N}+\Delta p_{N},m_{N}\rvert \theta_1)
    =&\sum_{i=1}^{N}\epsilon\big(v(\pi_{i}\rvert \theta_i)-v(\pi_{i-1}\rvert \theta_i)\big)+ O(\epsilon\Delta p)> 0
\end{split}
\end{equation}
where $\pi_{0}:=\pi_{N}$, and the strict inequality is due to our criterion that this be robust over a neighbourhood of $v$. Note that this also suggests $\pi_{N}\not\in \sigma(\theta_1)$.


This result inspires the following definition:
\begin{defn}[Graph Monotonicity\footnote{This concept resembles the closely related concept of cyclic monotonicity from mechanism design. In fact, when $G(\sigma)$ is connected, the equations of eq. \ref{eq:GCM} form a basis for cyclic monotonicty. The interested reader may consult Appendix \ref{app:CM} for a short discussion contrasting the two.}]
For a candidate equilibrium $\sigma\in\Sigma(u^0)$ with communication graph $G(\sigma)$,
we say that a modification $v$ is $G(\sigma)$\emph{-monotone} (or $v\in \GCM$) if, for any path $(\theta_1,\pi_1,\dots,\theta_{N},\pi_N=:\pi_{0})$ on $G(\sigma)$ with $N\ge 2$ we have
\begin{align}\label{eq:GCM}\tag{GM}
\sum_{i=1}^{N} v(\pi_{i}\rvert \theta_{i})-v(\pi_{i-1}\rvert \theta_{i})>0.
\end{align}
\end{defn}
Equation \ref{eq:Gintuit} suggests that this is a necessary condition for $\mathscr{O}$-stability. There is a question of whether this is a sufficient condition. The main challenge is that our argument relies on chains of indifferences to relate utilities. When the graph $G(\sigma)$ is not connected, there will be utilities that graph monotonicity cannot relate, even though the utility of one \emph{should} impose constraints on the other.

Recall that Theorem \ref{thm:gen} shows communication graphs (generically) must contain cycles to be non-connected. The following result argues that an acyclic communication graph is necessary for $\mathscr{O}$-stability, generically implying $\mathscr{O}$-stable equilibria are connected:
\begin{prop}\label{prop:acyc}
    $\GCM$ is non-empty iff the communication graph $G(\sigma)$ is acyclic.
\end{prop}
If $G(\sigma)$ is acyclic, we may construct a modification in $\GCM$ by making off-path messages unattractive. This can always be done through message-dependent modifications, and often through message-independent (ie. cheap talk) modifications as well. However, in cases where many equilibrium mixed actions have support within a small set of pure actions, obtaining graph monotonicity through cheap talk modifications may become complicated.

The converse is deduced by observing that eq. \ref{eq:GCM} for a path around a cycle changes sign if we flip the orientation of the path. Equivalently, observe that if $(\theta_1,\dots,\pi_N,\theta_1)$ is a cycle on $G(\sigma)$ then eq. \ref{eq:Gintuit} must be simultaneously greater than and less than zero.



Our main result of this section is the following, whose intuition we established in eq. \ref{eq:Gintuit}:
\begin{thm}\label{thm:main}
    For generic receiver utilities $u_R$, a candidate equilibrium $\sigma\in\Sigma(u^0,u_R)$ is $\mathscr{O}$-stable iff $\mathscr{O}\subseteq\GCM$.
\end{thm}
This relies on first showing that graph-monotonicity is a necessary property for $\mathscr{O}$-stability, thus ruling out cyclic candidate equilibria. We then apply Theorem \ref{thm:gen} to deduce that connectedness is generically a necessary condition for a candidate equilibrium to be $\mathscr{O}$-stable. The last step is to show that tree candidate equilibria are $\GCM$-stable.

\begin{proof}
    \textbf{Necessity:} We first show that graph-cyclic monotonicity is necessary for $\mathscr{O}$-stability. Suppose $v\not\in \GCM$. Then any neighbourhood of $v$ contains a modification $\tilde{v}$ that is not in the closure of $\GCM$ --- ie. for some path $(\theta_1,\pi_1,\dots,\theta_{N},\pi_N)$
\begin{equation}\label{sigfail}
    0>\sum_{i=1}^{N} \tilde{v}_i(\pi_{i})-\tilde{v}_i(\pi_{i-1})
\end{equation}
    where we adopt the shorthand $\tilde{v}_i(\pi):=\tilde{v}(\pi\rvert \theta_i)$. Since $\tilde{v}$ is continuous in $\pi$, this inequality holds in a neighbourhood $N_\pi$ of $\pi$.

    Now suppose our equilibrium can be approximated with actions $\hat{\pi}\rightarrow \pi$ as $\epsilon\rightarrow 0$. For this to be incentive compatible, we must have
\begin{align*}
    u(\hat{\pi}_i)-u(\hat{\pi}_{i-1})+\epsilon\left(\tilde{v}_i(\hat{\pi}_{i})-\tilde{v}_{i}(\hat{\pi}_{i-1})\right)=&0\quad \text{for }i\in\{2,\dots, N\}\\
    u(\hat{\pi}_N)-u(\hat{\pi}_1)+\epsilon(\tilde{v}_1(\hat{\pi}_N)-\tilde{v}_1(\hat{\pi}_1))\le& 0
\end{align*}
Summing up the first equation, and comparing with the second inequality, we get
\begin{equation*}
    u(\hat{\pi}_N)-u(\hat{\pi}_1)+\epsilon\left(\tilde{v}_1(\hat{\pi}_N)-\tilde{v}_1(\hat{\pi}_1)\right)\le 0=\sum_{i=2}^{N} u(\hat{\pi}_i)-u(\hat{\pi}_{i-1})+\epsilon\left(\tilde{v}_i(\hat{\pi}_{i})-\tilde{v}_i(\hat{\pi}_{i-1})\right).
\end{equation*}
Rearranging, we get
\begin{equation*}
    \tilde{v}_1(\hat{\pi}_N)-\tilde{v}_1(\hat{\pi}_1)\le\frac{u(\hat{\pi}_1)-u(\hat{\pi}_N)}{\epsilon}=\sum_{i=2}^{N} \tilde{v}_i(\hat{\pi}_{i})-\tilde{v}_i(\hat{\pi}_{i-1})
\end{equation*}
which contradicts eq. \ref{sigfail} when $\hat{\pi}\in N_\pi$.\medskip
\begin{figure}
    \centering
\begin{forest}for tree={
        grow'=east, 
        anchor=west, child anchor=west, 
        }
[$\pi^0$
    [$s^1_a$ [$\pi^2_a$ [$s^3_a$]] 
        [$\pi^2_b$ [$s^3_b$ [$\pi^4_a$[$s^5_a$]]]] 
        ]
    [$s^1_b$[$\pi^2_c$ [$s^3_c$] [$s^3_d$] ]]]
\end{forest}
    \caption{An example of a communication graph turned into a rooted tree. Our method is to approximate the children of $s^1_a$ (ie. $\pi^2_a,\pi^2_b,\pi^2_c$) before approximating further descendents of $s^1_a$ (ie. $\pi^4$). Graph monotonicity ensures that a state will never recommend a non-neighbouring action for slight modifications.}
    \label{fig:roottree}
\end{figure}
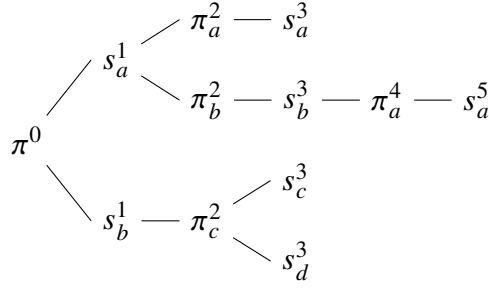

\textbf{Generic sufficiency:} From Proposition \ref{prop:acyc}, it suffices to check acyclic graphs. From Theorem \ref{thm:gen}, we can then further restrict attention to tree equilibria with one pure action. Fix the vertex associated with the pure action as the root of the communication graph. We adopt the following notation for node indices that we will maintain for related proofs: for a node $j$ on a rooted tree $(G,\pi_0)$: 
\begin{itemize}
    \item $j^\downarrow $ is the set of $j$'s children (ie. its neighbours that are furthest from the root)
    \item $j^\uparrow$ is the parent node of $j$ (ie. its unique neighbour closest to the root)
    \item $\mathcal{N}(j)=j^\downarrow\cup\{j^\uparrow\}$ is the set of neighbours of a node $j$
\end{itemize}
This notation will be pushed to subscripts to differentiate action/state nodes.

Let $\theta_j$ be a state and fix its parent action $\hat{\pi}_{j^\uparrow}$, for $\pi_k\in \pi_{j^\downarrow}$ let $I_k\subseteq \supp{\pi_k}$ be intervals such that
\begin{equation}\label{eq:neighIC.pure}
        \min_{\pi'_k\in I_k}u_S^{\epsilon v}(\pi'_{k}\rvert \theta_j)\le u_S^{\epsilon v}(\hat{\pi}_{j^\uparrow}\rvert \theta_j)
        \le \max_{\pi'_k\in I_k}u_S^{\epsilon v}(\pi'_{k}\rvert \theta_j)\qquad\text{for all }k\in j^\downarrow.
\end{equation}
Then by intermediate value theorem, for all children $k\in j^\downarrow$ there exist $\hat{\pi}_k\in I_k$ such that the sender in state $k$ is indifferent. For small $\epsilon$, and $\hat{\pi}_{j^\uparrow}\rightarrow \pi_{j^\uparrow}$, these inequalities can be satisfied with $I_k\rightarrow \{\pi_k\}$ and thus $\hat{\pi}_k\rightarrow \pi_k$. Applying this inductively down the tree gives us $\hat{\pi}\rightarrow \pi$.

Now we show that these neighbouring indifferences generate strict preferences for neighbouring actions as $\epsilon\rightarrow 0$. Suppose the actions are within a neighbourhood $N_\pi\ni\pi$ so that eq. \ref{eq:GCM} holds, then
\begin{equation*}
\begin{split}
    u(\hat{\pi}_1)-u(\hat{\pi}_N)+\epsilon\left(v_1(\hat{\pi}_1)-v_1(\hat{\pi}_N)\right)>&\sum_{i=2}^{N} u(\hat{\pi}_{i-1})-u(\hat{\pi}_{i})+\epsilon\left(v_i(\hat{\pi}_{i-1})-v_i(\hat{\pi}_{i})\right)\\
    =&\sum_{i=2}^{N}u_S^{\epsilon v}(\hat{\pi}_{i-1}\rvert \theta_i)-u_S^{\epsilon v}(\hat{\pi}_{i}\rvert \theta_i)= 0.\qedhere
\end{split}
\end{equation*}
\end{proof}
Proposition \ref{prop:acyc} and Theorem \ref{thm:gen} provide two immediate corollaries to this result:
\begin{cor}
\begin{enumerate}
    \item For generic $u_R$, a candidate equilibrium $\sigma\in\Sigma(u^0,u_R)$ is $\mathscr{O}$-stable for some non-empty set of modifications $\mathscr{O}$ iff its communication graph $G(\sigma)$ is acyclic.
    \item For generic $u_R$, there are finitely many candidate equilibria (up to information equivalence) that are $\mathscr{O}$-stable for some non-empty set of modifications $\mathscr{O}$.
\end{enumerate}
\end{cor}
\subsection{Reintroducing Idiosyncrasy}
\label{sec:idio}
\label{sec:Harstab}
To establish robustness of these equilibria, consider a sender preference with an idiosyncratic perturbation
$$u_S^{\epsilon V}(\pi\rvert \theta,\omega)=u^0(\pi)+\epsilon V(\pi\rvert \theta,\omega).$$

Note the determining factor for the distribution of messages sent from a state is the distribution of sender preferences in that state, ie. the conditional distribution of $V(\pi\rvert \theta,\cdot)$, and is independent of any correlation of preferences across states. Indeed, our assumption in Section \ref{sec:frag} that receiver and sender have independent preferences can be replaced with the assumption that $V(\pi\rvert \theta,\cdot)$ are identically distributed across states. To reflect this dependence on only the conditional distributions, we define the \emph{state-factored support} of an idiosyncratic perturbation $V$:
\begin{equation}\label{eq:supps}
    \supps{V}:=\bigtimes_{\theta\in \Theta}\supp{V(\pi\rvert \theta)}.
\end{equation}
Note that $\supp{V}\subseteq \supps{V}$ with equality when the state-dependent perturbations $\{V(\pi\rvert \theta)\}_\theta$ are mutually independent.\footnote{\label{fn:indep}This independence can always be obtained by reparametrizing idiosyncrasy --- if the idiosycrasy space is redefined as the product probability space $\tilde{\Omega}:=\Omega^\Theta$, the preference $\tilde{V}(\cdot\rvert \theta,\tilde{\omega}=(\omega_{\theta_1},\dots,\omega_{\theta_N})):=V(\cdot\rvert \theta,\omega_\theta)$ satisfies this independence and has the same conditional distributions as $V$. We find this parametrization unintuitive, as it prevents sender preferences from being correlated across states.}

To extend Theorem \ref{thm:main} to idiosyncratic perturbations, we make the following definition:
\begin{defn}\label{def:stab}
For a nonempty open set of modifications $\mathscr{O}\subseteq \mathbb{R}^{(\mathcal{A}\times M)\times \Theta}$, we say a candidate equilibrium $\sigma\in\Sigma(u^0)$ is \emph{Harsanyi $\mathscr{O}$-stable} if, for any compactly supported perturbation $V$ with $\supps{V}\subseteq 
\mathscr{O}$, there exists equilibria $\sigma^{\epsilon}\in \Sigma(u_S^{\epsilon V})$ to the game with sender preference $u_S^{\epsilon V}$ such that $\sigma^{\epsilon}\rightarrow \sigma$ as $\epsilon\rightarrow 0$.
\end{defn}
Note that Harsanyi stability is equivalent to Harsanyi $\mathbb{R}^{(\mathcal{A}\times M)\times \Theta}$-stability. Moreover, Harsanyi $\mathscr{O}$-stability is a stronger condition than $\mathscr{O}$-stability. Nevertheless, our result ends up closely mirroring Theorem \ref{thm:main}:
\begin{thm}\label{thm:mainB}
    A tree candidate equilibrium $\sigma\in \Sigma(u^0)$ is Harsanyi $\mathscr{O}$-stable iff $\mathscr{O}\subseteq \GCM$.

    Suppose $v\in \GCM$ is a non-idiosyncratic modification and $V_n\rightarrow v$ uniformly, then there exist $N,\overline{\epsilon}>0$, and equilibria $\sigma^{\epsilon V_n}\in\Sigma(u_S^{\epsilon V_n})$ and $\sigma^{\epsilon v}\in\Sigma(u_S^{\epsilon v})$ such that
    \begin{align*}
         (1)&\qquad \sigma^{\epsilon V_n}\rightarrow \sigma^{\epsilon v}\qquad \text{as }n\rightarrow \infty,\text{ for }\epsilon<\overline{\epsilon}, \text{ and}\\
         (2)&\qquad \sigma^{\epsilon V_n},\sigma^{\epsilon v}\rightarrow \sigma\qquad \text{as }\epsilon\rightarrow 0,\text{ for }n>N.
    \end{align*}
\end{thm}
The first statement of the theorem says that a slight idiosyncratic perturbation in the direction $\GCM$ will allow us to approximate the candidate equilibrium $\sigma$. The second statement says that the approximating equilibria that we obtained in Theorem \ref{thm:main} are Harsanyi-stable, note that the limits (2) are an immediate consequence of the first statement.\footnote{We expect this is a complete description of Harsanyi stable approximating equilibria when either (1) $V$ is action-independent, (2) $\sigma$ is acyclic, or (3) $|\mathcal{A}|=3$. Otherwise, there may be cyclic candidate equilibria that can be approximated with Harsanyi stable equilibria --- see \cite{SGGK23} for analysis in two-state models. In Appendix \ref{app:weakstab} the distinction between Harsanyi stable approximating equilibria and Harsanyi $\mathscr{O}$-stability is discussed in the context of an example.
    }

The proof uses the same rooted tree as in the proof of Theorem \ref{thm:main}. The goal is to prove that senders in a given state can be induced to send each neighbouring message (on $G(\sigma)$) with the correct probability. To simplify this, we temporarily ignore the possibility of sending non-neighbouring messages, and anticipate our argument moving inductively down the tree from the root with fixed (pure) action. Thus at each state, its parent action is fixed and we seek to perturb its child actions induce the desired sender strategy.

The following lemma provides a sufficient condition for this to be obtainable within a neighbourhood of equilibrium actions. For a message-action $\pi=(p,m)$ we denote $\supp{\pi}=\{(a,m);a\in \supp{p}\}$. 
\begin{lem}[Neighbour Incentive Compatability]\label{lem:NIC}
Let $u_S:\Delta\mathcal{A}\times M\times\Omega\rightarrow \mathbb{R}$ be some idiosyncratic sender preference, $\mathcal{N}$ be a finite set of message-actions containing at most one pure action, and $m\in\Delta\mathcal{N}$ be a mixed sender strategy over this set. If $\mathcal{N}$ includes a pure action, denote it by $\pi_0$, otherwise fix $\pi_0\in\mathcal{N}$ an arbitrary message-action. For every other action $\pi_j\in \mathcal{N}\setminus\{\pi_0\}$, let $ B_j\subseteq \Delta\supp{\pi_j}$ be closed, convex intervals that satisfy 
\begin{align}\label{eq:NIC.bound}
    \inf_{\pi'\in B_j}\P[\omega]{u_S(\pi')>u_S(\pi_{0})}=&0&
    \sup_{\pi'\in B_j}\P[\omega]{u_S(\pi')>u_S(\pi_{0})}=&1.
\end{align}
There then exists a profile of actions $\hat{\pi}_+\in \bigtimes_{j\neq 0} B_j$ that induce the sender population best response is $m$ when limited to messages in $\mathcal{N}$.
\end{lem}
This says that, fixing the state, we can induce any sender strategy (over neighbouring messages) whenever the parent action $\pi_0$ is unanimously ranked in the interior of each neighbourhood $B_j$ of child actions.

To conclude the proof, we verify that small neighbourhoods $B_j$ satisfying eq. \ref{eq:NIC.bound} can simultaneously be chosen for every message-action in $\sigma(S)$, then apply Theorem \ref{thm:main} to show that no sender will deviate to non-neighbouring messages.

\section{Ordinally Transparent Preferences}
So far we have limited our analysis to approximately transparent preferences, providing a type of continuity result, but in reality agents' preferences are always a finite distance away from transparency, and may exhibit significant state dependence without fundamentally disturbing our equilibria from the previous section. 

In this section we consider a broad class of preferences that we call \emph{ordinally transparent}, and provide a sufficient condition on the preference's state dependence to admit persuasive equilibria. This novel condition bridges graph monoticity with more traditional single crossing conditions.

\begin{defn}[Ordinal Transparence]
    We say that a sender's preference $u_S$ is \emph{ordinally transparent} if sender's have a strict unanimous ranking over all $(a',m'), (a,m)$ with $a'\neq a$:
    \begin{equation}\label{eq:strong}\begin{split}
        \P{u_S(a,m\rvert \theta,\omega)=u_S(a',m'\rvert \theta,\omega)}=&0\\
        \P{u_S(a,m\rvert \theta,\omega)>u_S(a',m'\rvert \theta,\omega)}\in&\{0,1\}
    \end{split}
    \end{equation}
    The utility $u^0:\Delta\mathcal{A}\times M\rightarrow \mathbb{R}$ is a \emph{transparent representation} of $u_S$ if
    $$
    u^0(a,m)>u^0(a',m')\qquad\Leftrightarrow\qquad \P{u_S(a,m\rvert \theta,\omega)>u_S(a',m'\rvert \theta,\omega)} =1
    $$
\end{defn} 
We exclude $a=a'$ in our definition of ordinal transparence to include cheap talk in our analysis, paralleling Assumption (S).

While transparent preferences are public and state-independent, ordinally transparent preferences may exhibit idiosyncrasy and state-dependence to a moderate degree, so long as the ordinal ranking of pure outcomes is transparent. Within cheap talk, this amounts to having a strict, transparent ordinal preference over $\mathcal{A}$ --- however the risk attitudes can vary significantly. Any ordinally transparent preference can be decomposed into its transparent representation $u^0$ plus a perturbation that preserves this ordinal ranking.

The \cite{IK87} Beer-Quiche Game provides a classic example of ordinally transparent preferences, where depending on the state (the sender being `weak' or `strong'), senders have slightly different preference over messages (having quiche or beer), but they have a transparent preference over the receiver's action (preferring to avoid a duel). This game admits a persuasive `partial pooling' equilibrium when the sender is weak with high probability, or if a message can be accompanied by burning money.

To compare equilibria within this broad class of preferences, we use informational equivalence:
\begin{defn}
    We say two strategy profiles $\sigma,\sigma'$ are \emph{informationally equivalent} if they induce the same distribution of posterior beliefs in the receiver.
\end{defn}
Two informationally equivalent equilibria may differ in the receiver's choice of mixed action or by permutating the messages. In this section only the former effect --- the receiver adjusting their action to generate sender indifference --- is relevant.

Our analysis of ordinally transparent preferences begins by observing that their equilibria are represented by our candidate equilibria from earlier sections:
\begin{prop}\label{prop:translucent}
    If the sender has ordinally transparent preference $u_S$, then any equilibrium $\sigma$ is informationally equivalent to a candidate equilibrium $\sigma'\in\Sigma(u^0)$, where $u^0$ is a transparent representation of $u_S$.
\end{prop}
From an information perspective the equilibria of a model with ordinally transparent preferences are a subset of candidate equilibria. This allows our analysis in Section \ref{sec:receiver}, studying the generic structure of candidate equilibria, to extend to describe the equilibria of models with ordinally transparent preferences.

\label{sec:bigeps}
Naturally, we seek to answer which of these equilibria will be admitted by such a model. A simple exercise will show that in two-state models, any equilibrium for a small perturbation will also be preserved (up to informational equivalence) as $\epsilon$ increases, so long as the preference remains ordinally transparent. Thus it is sufficient that $u_S$ is graph monotone. This is not true in general, however we can provide a stronger sufficient condition.

First, we observe that graph monotonicity is implied by the following condition, which we term \emph{strong graph-monotonicity}\footnote{We speculate that this concept may have applications in mechanism design. See Appendix \ref{app:CM} for details.} ($\SLID$):

\begin{prop}\label{prop:sandwich} 

If $v$ satisfies
\begin{equation}\tag{SGM}
    v(\pi_0\rvert \theta_0)-v(\pi_N\rvert \theta_0)\ge v(\pi_0\rvert \theta_1)-v(\pi_N\rvert \theta_1)
\end{equation}
for any path $(\theta_0,\pi_0,\dots,\theta_{N},\pi_N)$ on $G(\sigma)$, with a strict inequality when $N=1$, then $v\in \GCM$.
\end{prop}
Strong graph-monotonicity can be seen as a `semi-local' version of single-crossing. For comparison, a graph version of single crossing requires that for any path $(\theta_0,\dots,\pi_N)$:
$$
v(\pi_{n_1}\rvert \theta_{m_1})-v(\pi_{n_2}\rvert \theta_{m_1})\ge v(\pi_{n_1}\rvert \theta_{m_2})-v(\pi_{n_2}\rvert \theta_{m_2})\quad\text{for all }n_1<n_2,\,m_1<m_2,
$$
ie. states have a relative preference for messages on `their side' of the graph. 

While single-crossing requires this relative preference to hold relative to all other states on the path, for every action pair; graph monotonicity fixes one action to be an on-path action, and only requires a preference relative to the neighbouring state.

The following condition ensures that strong graph-monotonicity will hold for all actions sharing the same support:
\begin{defn}[Semi-Local Increasing Differences]
    We say that a modification $v$ is (support) $\sigma$-\emph{semi-local increasing  in differences} (or $v\in \SLIDS$) if, for any path $(\theta_0,\pi_0,\dots,\theta_{N},\pi_N)$ on $G(\sigma)$, we have
    \begin{align}\label{SLID}\tag{SLID}
     v(\alpha_0\rvert \theta_0)-v(\alpha_N\rvert \theta_0)\ge v(\alpha_0\rvert \theta_1)-v(\alpha_N\rvert \theta_1)
    \end{align}
    for all pure message-actions $\alpha_i\in\supp{\pi_i}$, and when $N=1$ there is strict inequality for at least one $(\alpha_0,\alpha_1)$ permutation.
\end{defn}
Note if $v\in\SLIDS$ then $u^0+ v\in\SLIDS$ for any transparent preference $u^0$. Thus support-monotonicity is independent of how a preference is decomposed into modification and transparent preference.

Moreover $\SLIDS$ is non-empty whenever $G(\sigma)$ is acyclic --- this can be observed simply by inductively making off-path messages unattractive.

It is clear $\SLIDS\subseteq \SLID\subseteq \GCM$. When a modification is action-independent, the permutations of pure-actions in the definition of support-monotonicity are irrelevant, and $\SLIDS=\SLID$. If there further exists a state $s^\ast$ sending every equilibrium message, then $\SLIDS=\GCM$.

We can also verify that $\SLIDS=\GCM$ for tree equilibria involving three pure actions --- since these equilibria must have two messages, there is a state sending every message, and after normalizing the values of the randomized action's support, any modification is equivalent to a message-based modification.
\begin{thm}\label{thm:largeeps}
    Let $u_S$ be an ordinally transparent preference and $\sigma\in\Sigma(u^0)$ be a candidate equilibrium for an associated transparent preference $u^0$. If $\supps{u_S}\subseteq\SLIDS$ then this preference admits an equilibrium $\hat{\sigma}\in\Sigma(u_S)$ that is informationally equivalent to $\sigma$.


\end{thm}This result shows that perturbing an ordinally transparent preference $u_S\in \SLIDS$ by a modification $v\in\SLIDS$, preserves the equilibrium $\sigma$ so long as ordinal transparence is maintained.
\begin{proof}
\textbf{(1) Obtain strict preferences (without idiosyncrasy):} Suppose senders in each state are indifferent over the messages that neighbour this state on $G(\sigma)$ (which we will later prove using these intermediate results). We prove they then have strict preference for neighbouring messages on $G(\sigma)$ over non-neighbouring messages by induction on the $G(\sigma)$-distance between the state and the non-neighbouring message. 
\begin{figure}
    \centering
    \begin{subfigure}{0.48\textwidth}\centering
    \begin{tikzpicture}
        \filldraw (3.5,3.5) circle (1.5 pt) node[anchor = west] {\color{black}$\overline{\alpha}_N$};
        \filldraw (0,0) circle (1.5 pt) node[anchor = east] {\color{black}$\underline{\alpha}_N$};
        \filldraw (4,4.3) circle (1.5 pt) node[anchor = south] {\color{black}$\overline{\alpha}_{N-1}$};
        \filldraw (1,1.9) circle (1.5 pt) node[anchor = north] {\color{black}$\underline{\alpha}_{N-1}$};
        \draw (0,0) -- (3.5,3.5);
        \draw (1,1.9) -- (4,4.3);
        \draw[dashed, blue, thick] (2,0) -- (2,5);
        \node[anchor = east] at (0,5) {$u_S(\cdot \rvert \theta_{N-1})$};
        \node[anchor = north] at (5,0) {$u_S(\cdot \rvert \theta_{N})$};
        \filldraw[red] (2,2) circle (1.5 pt) node[anchor = north west] {\color{black}$\hat{\pi}_N$};
        \filldraw[red] (2,2.7) circle (1.5 pt) node[anchor = south east] {\color{black}$\hat{\pi}_{N-1}$};
        \draw[<->] (0,5) -- (0,0) -- (5,0);
    \end{tikzpicture}
    \caption{}
    \label{fig:largeeps.base}
    \end{subfigure}
    \begin{subfigure}{0.48\textwidth}\centering
    \begin{tikzpicture}
        \draw (0,0) -- (3.5,3.5);
        \draw (1,2.4) -- (4,4.8);
        \draw[dashed, thick, blue] (3,0) -- (3,5);
        \node[anchor = north] at (3,0) {\color{black}$\hat{\pi}_{k+1}$};
        \filldraw (3.5,3.5) circle (1.5 pt) node[anchor =  west] {\color{black}$\overline{\alpha}_N$};
        \filldraw (0,0) circle (1.5 pt) node[anchor =  east] {\color{black}$\underline{\alpha}_N$};
        \filldraw (4,4.8) circle (1.5 pt) node[anchor = west] {\color{black}$\overline{\alpha}_{k}$};
        \filldraw (1,2.4) circle (1.5 pt) node[anchor = north] {\color{black}$\underline{\alpha}_{k}$};
        \filldraw[red] (3,4) circle (1.5 pt) node[anchor =  south east] {\color{black}$\hat{\pi}_{k}$};
        \node[anchor = east] at (0,5) {$u_S(\cdot \rvert \theta_{k})$};
        \node[anchor = north] at (5,0) {$u_S(\cdot \rvert \theta_{k+1})$};
        \filldraw[red] (2,2) circle (1.5 pt) node[anchor = north west] {\color{black}$\hat{\pi}_N$};
        \draw[<->] (0,5) -- (0,0) -- (5,0);
    \end{tikzpicture}
    \caption{}
    \label{fig:largeeps.induct}
    \end{subfigure}
    \caption{Illustration of the proof of Theorem \ref{thm:largeeps}. The graphs compare the utilities attained by senders in two states $\theta_{k},\theta_{k+1}$ separated by a single message $\pi_k$ on $G(\sigma)$ for various mixed actions (on the left $k=N-1$). SLID implies that $\Delta\supp{\pi_k}$ will always lie above the affine hull $\text{aff}(\supp{\pi_{N}})$ on these graphs. The indifferences of the sender in state $\theta_{k+1}$ are given by vertical lines.
    }
    \label{fig:largeeps}
\end{figure}
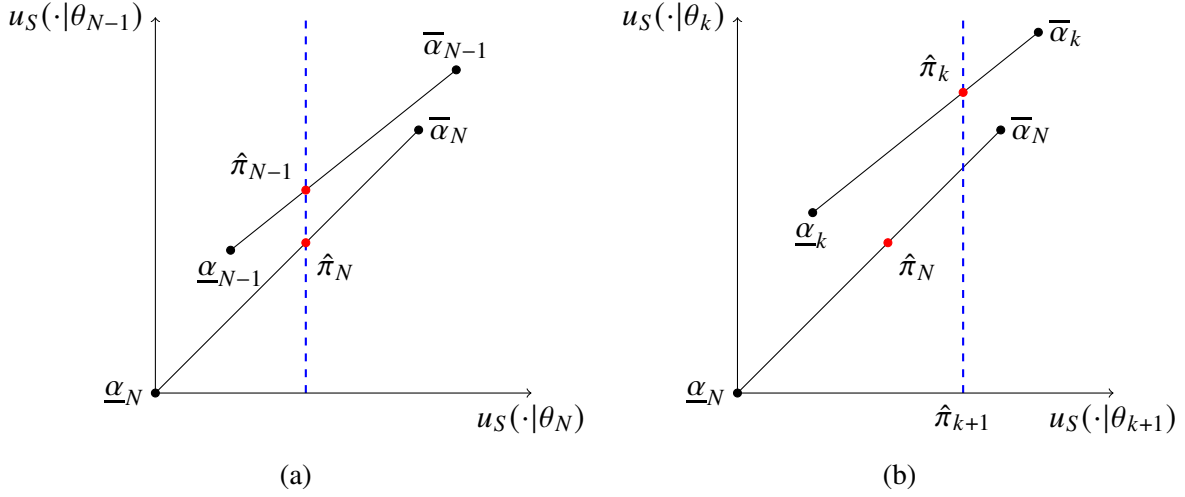

In the base case (illustrated in Figure \ref{fig:largeeps.base}), the connecting path is $(\theta_{N-1},\pi_{N-1},\theta_N,\pi_N)$. We know $\hat{\pi}_{N-1}$ will be chosen so that the sender in state $\theta_N$ is indifferent between $\hat{\pi}_{N-1}$ and $\hat{\pi}_N$, but by SLID this implies that the sender in state $\theta_{N-1}$ will strictly prefer $\pi_{N-1}$ over $\pi_N$.

Now suppose that the path is $(\theta_k,\pi_{k},\dots,\pi_N)$ (illustrated in Figure \ref{fig:largeeps.induct}). Again the action $\hat{\pi}_k$ is such that the sender in state $\theta_{k+1}$ is indifferent between $\hat{\pi}_k$ and $\hat{\pi}_{k+1}$. By the induction hypothesis, the sender in state $\theta_{k+1}$ strictly prefers $\hat{\pi}_{k+1}$ (and therefore $\hat{\pi}_k$) to $\hat{\pi}_N$. By SLID, the sender in state $\theta_k$ has the same strict preference for $\hat{\pi}_{k}$ over $\hat{\pi}_N$.

\textbf{(2) Neighbouring Incentive Compatibility:} For this step denote $\underline{\alpha}_k:=\min \supp{\pi_k}$, $\overline{\alpha}_k:=\max \supp{\pi_k}$. We root our tree at its pure action $\alpha_0$\footnote{If no actions are pure, set $\alpha_0$ to either be the lowest upper support $\argmin u_S(\overline{\alpha})$ or highest lower support $\argmax u_S(\underline{\alpha})$  over equilibrium message-actions.}
, as in the proof of Theorem \ref{thm:main}. Note that as a candidate equilibrium, we must have $\underline{\alpha}_k<\alpha_0<\overline{\alpha}_k$ for every non-root message $\pi_k\in\sigma(\Theta)$.

Proceeding inductively down the tree, suppose state $\theta_k$ is the first state where we are unable to obtain sender indifference over neighbouring actions, ie. for some child $\pi_{k^\downarrow}$ of $\theta_k$ either
\begin{equation}\label{eq:nolem}
    (1)\quad u_S(\underline{\alpha}_{k^\downarrow}\rvert \theta_k)>u_S(\hat{\pi}_{k^\uparrow}\rvert \theta_k)
    ,\qquad\text{or}\qquad 
    (2)\quad u_S(\overline{\alpha}_{k^\downarrow}\rvert \theta_k)<u_S(\hat{\pi}_{k^\uparrow}\rvert \theta_k).
\end{equation}
In either case, we may apply Step 1 to obtain incentive compatibility on the tree above $\theta_k$, obtaining
\begin{equation}\label{eq:ordinalviol}
\begin{split}
    (1)\qquad& u_S(\alpha_0\rvert \theta_k)\le u_S(\hat{\pi}_{k^\uparrow}\rvert \theta_k)< u_S(\underline{\alpha}_{k^\downarrow}\rvert \theta_k), \text{ or}\\
    (2)\qquad& u_S(\alpha_0\rvert \theta_{1})\ge u_S(\hat{\pi}_{k^\uparrow}\rvert \theta_{1})>u_S(\overline{\alpha}_{k^\downarrow}\rvert \theta_{1}),
\end{split}
\end{equation}
respectively (where the first equalities holds iff $\pi_{k^\uparrow}=\alpha_0$). Both cases violate ordinal transparency.

\textbf{(3) Extension to Idiosyncratic Preferences:} We follow the same steps as in the non-idiosyncratic proof. Suppose there exists actions $\pi$ that induce the desired degree of randomization when constrained to neighbouring messages. To observe a sender with realized preference $u_S(\cdot\rvert \theta_0,\omega)\in\supp{u_S(\cdot\rvert \theta_0)}$ will have strict preferences for neighbouring messages, consider a sequence $(\theta_0,\dots,\pi_N)$. Applying intermediate value theorem, we find a utility $\hat{u}_S\in \text{co}(\supps{u_S})\subseteq \SLIDS$ that is indifferent over each message is this sequence. Now apply the previous induction result to the preference $\overline{u}_S\in\SLIDS$ given by
$$
\overline{u}_S(\pi\rvert s):=\begin{cases}
    \hat{u}_S(\pi\rvert \theta)&\theta\neq \theta_0\\
    u_S(\pi\rvert \theta_0,\omega)& \theta=\theta_0.
\end{cases}
$$

It remains to show that we can find mixed actions that induce sender's to send each of their neighbouring messages with the appropriate probability, when limited to neighbouring actions. Working inductively down the tree, as long as eq. \ref{eq:nolem} is almost never satisfied at a state, we can apply Lemma \ref{lem:NIC} to find appropriate child actions $\hat{\pi}_{k^\downarrow}\in \supp{\pi_{k^\downarrow}}$ generating the desired distribution of messages sent from a state. 

When eq. \ref{eq:nolem} fails with positive probability, we are able to obtain a utility $\overline{u}_S\in\text{co}(\supps{u_S})$ satisfying eq. \ref{eq:ordinalviol}, violating ordinal transparency.
\end{proof}

We will see that this result can be applied to many natural modifications that we consider in Section \ref{sec:apps}.

\section{Applications}\label{sec:apps}
\label{sec:mods}
In this section we apply the above results to three natural modifications. Two modifications represent ethical considerations of a sender, in one case concerned about the receiver's welfare, in the second concerned about lying.

Ethical modifications intuitively seem to favour communication, however we will see that this is not the case in general. We will provide a natural example illustrating how vertically differentiated states may be unfavourable to communication when the sender is empathetic. 

We will apply our lying aversion modification to illustrate how money burning can benefit the sender's persuasive capabilities when they have nearly transparent preferences, despite not improving the sender's utility of candidate equilibria.

The last modification we will consider is a technological one, resembling a `weak disclosure game'. As in standard disclosure games the sender reveals information to the receiver, however they also have the option to cheaply fabricate false information. This technology is of particular interest, as it is maximally stabilizing modification.

Since we are discussing fixed modifications it will be helpful to adopt the shorthand of saying a modification $v$ \emph{stabilizes} a candidate equilibrium $\sigma\in\Sigma(u^0)$ when $\sigma$ is $\mathscr{O}$-stable for some neighbourhood $\mathscr{O}\ni v$.
\subsection{Other-regarding Senders}
\label{sec:emp}

When faced with a decision that has little effect on their material utility, even the most egoistic sender may consider the effect their actions will have on others as a deciding factor. We consider two modifications of this flavour: empathy ($v_E:=u_R$) and antipathy ($v_A:=-u_R$).

These modifications correspond to aligning or disaligning preferences. A common rule-of-thumb in communication games is that preference alignment increases the possibilities of communication: since the receiver will choose the best option according to their belief, a sender with a stronger interest in the receiver will seek to ensure the receiver makes `more accurate' decisions.


The previous sections illustrate that this intuition is validated when the communication structure and the receiver's preference share compatible geometries. We illustrate with using the special case where communication and the receiver's preference are `linearly ordered' in compatible manners. In particular, let $\succ$ be a linear order on the state space $\Theta$. 
\begin{defn}
    We say the receiver utility $u_R$ is $\succ$-\emph{single crossing} if     
    \begin{equation*}
        \theta\mapsto u_R(a\rvert \theta)-u_R(a'\rvert \theta)
    \end{equation*}
    is $\succ$-strictly monotone for all $a,a'\in\mathcal{A}$.

    A candidate equilibrium $\sigma\in\Sigma(u^0)$ is an $\succ$-\emph{interval} equilibrium if the sender's strategy is $\succ$-ordered: 
    for $\pi,\pi'\in\sigma(\Theta)$ either
    \begin{align*}
        \max_\succ\sigma^{-1}(\pi)\preceq& \min_\succ\sigma^{-1}(\pi')&
        \text{or}&&
        \min_\succ\sigma^{-1}(\pi)\succeq& \max_\succ\sigma^{-1}(\pi').
    \end{align*}
\end{defn}
The single crossing condition induces a linear order $\succ_\mathcal{A}$ on the action space $\mathcal{A}$ such that $\succ_\mathcal{A}$-higher actions have a comparative advantage in $\succ$-high states.

Interval equilibria describe situations where there are $\succ$-thresholds between messages: only lower states send one message, and only higher states send the other. Such equilibria are necessarily acyclic.

Combining these structures ensures that empathy will stabilize equilibria:

\begin{prop}\label{prop:emp}
If the sender's utility is $u_S:=u^0-\epsilon u_R$ for some transparent $u^0$ and $\epsilon>0$, then there is no persuasive equilibrium in $\Sigma(u^0-\epsilon u_R,u_R)$ .

Let $u^0$ be a transparent sender preference, and $u_R$ be a $\succ$-single crossing receiver preference. If $\sigma\in\Sigma(u^0,u_R)$ is an $\succ$-interval candidate equilibrium, then there is an informationally equivalent equilibrium $\sigma'\in \Sigma(u^0+\epsilon u_R,u_R)$ whenever $\epsilon>0$ is sufficiently small that the sender preference $u^0+\epsilon u_R$ is ordinally transparent.
\end{prop}


The first result is similar to the impossibility of communication in zero-sum games (see \cite{MV78}). This shows that when the communication technology is cheap talk/money burning, this impossibility extends to games that are even slightly antagonistic\footnote{This only works from the baseline of transparent preferences. In Example \ref{ex:tri} we find a case where adding slight antagonism to a state-\emph{dependent} preference \emph{increases} communication possibilities.} --- even if revelation is pareto improving as in Example \ref{ex:bi}. This result does not use the techniques developed in this paper, extending beyond ordinally transparent preferences and even applying in infinite state/action spaces. 

The second result can be seen as an analog to \cite{CS82} style cheap talk, where increasing alignment weakly improves communication. Here the result is that increasing alignment preserves communication --- at least until the preference is no longer ordinally transparent, at which point more informative non-connected equilibrium may emerge. 

The assumptions for the second result are automatically satisfied in two state models, and in the following example whose candidate equilibria are more fully analyzed in \cite{LR20}:
\begin{ex}[Advice on Investing in an Asset]\label{ex:asset}
    An investor consults a broker about an asset, after which they decide what share of their wealth to invest in the asset. The action space is thus $\mathcal{A}=[0,1]$,\footnote{This action space is continuous. Because the utilities (in particular $u^\ast$) are well behaved our results still apply, alternatively the reader may assume a finite approximation of this space.} 
    with the investor initially holding a position $a_0\in [0,1]$. The broker is aware of some information $\theta\in \Theta$ (drawn from the prior $\mu$), indicating the investor's ideal position is actually $a^\ast(\theta)\in [0,1]$, and receives a fee proportional to the investor's trade volume paid by the investor. The parties' material preferences are thus
    \begin{align*}
        u_R(a\rvert \theta)=&-\tfrac{1}{2}(a-a^\ast(\theta))^2-\kappa|a-a_0|&
        u^0(a)=&|a-a_0|&
    \end{align*}
    We assume that the investor's initial position is optimal under their prior knowledge, $a_0=\E[\mu]{a^\ast(\theta)}$.

    States can be ordered by their optimal positions, in which case the investor's preference satisfies single-crossing. Moreover, the sender-preferred candidate equilibrium is attained through a cut-off policy reflecting whether or not the optimal investment is above a threshold $a^\ast(\theta^\ast)$ (where $\theta^\ast$ randomizes over the two messages). This is an interval equilibrium, ensuring empathy is stabilizing.
\end{ex}

However there are simple three state models where empathy is not $G(\sigma)$-monotone for any persuasive candidate equilibrium $\sigma\in\Sigma(u^0)$. Consider the following example:
\begin{ex}[Salesperson with Vertically Differentiated Products]\label{ex:tri} 
A lawmaker is considering a policy to enact, and faces three alternatives: pass the risky policy $A$
A consumer seeks to purchase a product, and faces three alternatives: buy the cheaper value-brand product $A$, buy the expensive premium-brand product $B$, or purchase nothing $\emptyset$. 

The consumer does not know which product is suitable for their purposes. They have two concerns in particular: (1) the product in general may not be appropriate for their uses, this is state $n$; (2) the quality of the value-brand item. It may be of similar quality to the premium brand, in which case the state is $a$; or it may be the case that the premium-brand version is higher quality, which is indicated by state $b$.

The salesperson is aware of the state and communicates through cheap talk.

This leads the receiver to have the following preferences motivated below, the salesperson has transparent material incentives based on commission:
\begin{center}
    \begin{tabular}{r|c|c|c}
         $_\theta\smallsetminus ^\alpha$&$A$&$\emptyset$&$B$ \\\hline
         $u_R(\,\cdot\,\rvert a)$& 5/4&3/4&1\\
         $u_R(\,\cdot\,\rvert b)$&0&3/4&1\\
         $u_R(\,\cdot\,\rvert n)$&0&3/4&$-2$
    \end{tabular}\hspace{2cm}
    \begin{tabular}{r|c|c|c}
         &$A$&$\emptyset$&$B$  \\\hline
         $u^0$& 1/2&0&1
    \end{tabular}
\end{center}
From the consumer's perspective, the value-brand item $A$ is only useful in state $a$, where it has more value than item $B$ through its cheaper cost; in other states any value provided by product $A$ is offset by its cost. The premium-brand item $B$ delivers value in states $a,b$, but in state $n$ its value is significantly overshadowed by its cost. 

These preferences lead to the indirect utility illustrated in Figure \ref{fig:uncoopemp}.

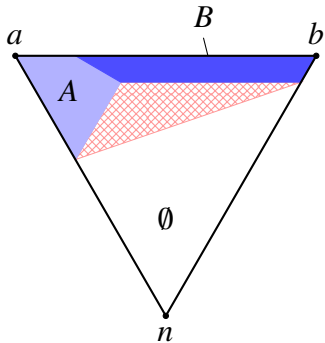
\begin{figure}\centering\begin{tikzpicture}
    \coordinate (a) at (0,3.46);
    \coordinate (b) at (4,3.46);
    \coordinate (n) at (2,0);
    \coordinate (AN) at (0.8,2.08);
    \coordinate (ABN) at (1.4,3.11);
    \coordinate (AB) at (0.8,3.46);
    \coordinate (BN) at (3.8, 3.11);
    \draw (2.6,3.3) -- (2.5,3.7);
    \node[anchor = south] at (2.5,3.7) {$B$};
    \node[anchor = south] at (2,1) {$\emptyset$};
    \draw[pattern=crosshatch, pattern color=red!40!white, draw=red!40!white, ultra thin] (AN) -- (ABN) -- (BN) -- cycle;
    \filldraw [blue!30!white] (a) -- (AB) -- (ABN) -- (AN) -- cycle; 
    \filldraw [blue!70!white] (b) -- (AB) -- (ABN) -- (BN) -- cycle; 
    \draw[thick] (a) -- (b) -- (n) -- cycle;
    \node at (0.7,3) {$A$};
    \filldraw (0,3.46) circle (1 pt) node[anchor = south] {\color{black}$a$};
    \filldraw (4,3.46) circle (1 pt) node[anchor = south] {\color{black}$b$};
    \filldraw (2,0) circle (1 pt) node[anchor = north] {\color{black}$n$};
    \end{tikzpicture}
    \caption{Shaded regions indicate the sender's indirect utility in Example \ref{ex:tri}, with darker hues indicating higher utility. There are persuasive candidate equilibria when the prior belief is in the cross-hatched region. However no persuasive equilibrium survives the introduction of empathy into the sender's preference. }
    \label{fig:uncoopemp}
\end{figure}

\begin{prop}[Uncooperative empathy]\label{prop:antipath}
    Suppose the sender has the ordinally transparent preference $u+\epsilon u_R$, where $\epsilon>0$, then there is no persuasive equilibrium for any prior $\mu$. 
    
    If the prior is in the cross-hatched region there are persuasive candidate equilibria $\sigma\in\Sigma(u^0)$ such that $\SLIDS$ is non-empty and excludes empathy.
\end{prop}
The key property of our example for this proposition is that the consumer much prefers $A$ over the candidate equilibrium action $p_B:=\frac{1}{2}(\emptyset\oplus B)$ in state $n$. As such an empathetic salesperson in state $n$ will seek to break ties in favour of action $A$ rather than the potentially harmful action $p_B$ (specifically, relative to a salesperson in state $a$). This is contrary to persuasive acyclic candidate equilibria however, which require that the salesperson sends a message that induces the action $p_B$ in state $n$, but not in state $a$.

Empathy can even be harmful to communication in this setting. Suppose the prior belief is in the cross-hatched region of Figure \ref{fig:uncoopemp}, and consider an acyclic persuasive candidate equilibrium $\sigma$ and a slight modification $ v\in\GCM$ that permits persuasion. Adding a relatively large amount of empathy will shift the modification outside of the set $\GCM$, removing the possibility of persuasion. In this way, increasing preference alignment may decrease the receiver's expected utility. Undoing this alignment provides an example where adding antipathy increases communication possibilities.\footnote{It is essential that the preference to which we add antipathy is already state dependent, to avoid the domain of Proposition \ref{prop:emp}.}
\end{ex}



\subsection{Lying Aversion}\label{sec:nearct}
A separate class of modifications are those that only depend on the message sent. This independence of the receiver's action means that stability of a candidate equilibrium can often be directly inferred from its communication graph. We explore this by applying our analysis to a lying averse sender.

Consider a sender concerned with honesty, opting to break their indifferences in favour of the truth. To model this we consider a message space identified with the state space through a bijection $\psi:M\leftrightarrow \Theta$ where a message $m$ is the claim `The state is $\psi(m)$.' The sender views the ethical cost of lying through the following modification
\begin{align}\label{eq:LA}
    v_{\text{LA}}(m\rvert \theta):=&\begin{cases}
    0&\psi(m)=\theta\\
    \ell_{\theta}&\psi(m)\neq \theta,
    \end{cases}
\end{align}
where $\ell_\theta<0$ for all $\theta$, ie. all lies from a state $\theta$ are equally costly.\footnote{The reader may consider this a benchmark, different results may be obtained if the cost of lying is anisotropic. 

Alternate models of lying aversion use a message-independent mechanism (see \cite{CD06}), where `lying' imposes a psychological cost through the guilt of inducing a misleading belief in the receiver. \cite{V08} experimentally tests these two models and finds that lying aversion is better described as a ``preference for keeping [one's] word'', as might be described by Equation \ref{eq:LA}. 

\cite{K09} applies a more general, but similar, model of lying aversion to Crawford-Sobel cheap talk.}

While lying aversion intuitively seems to favour honesty, it is actually contrary to the structure of candidate equilibria which generically require obfuscating the state to some degree (as we saw in Theorem \ref{thm:gen}). Thus lying is necessarily part of equilibria, and the challenge with lying averse senders is ensuring that they are only incentivize to send the `right' lies. This significantly constrains communication.

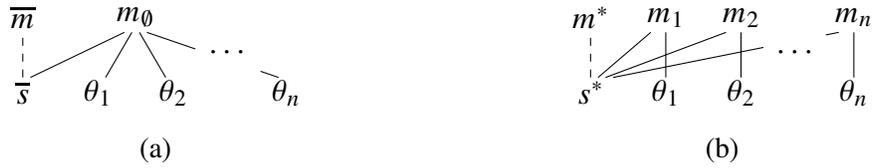
\begin{figure}
\centering
    \begin{subfigure}{0.45\textwidth}\centering
    \begin{tikzpicture}
    \node at (0,0) {$\overline{s}$};
    \node at (0,1) {$\overline{m}$};
    \node at (1.5,1) {$m_{\emptyset}$};
    \draw[dashed] (0,0.2) -- (0,0.8);
    \draw (0.1,0.2) -- (1.35,0.8);
    \node at (1,0) {$\theta_1$};
    \draw (1.1,0.2) -- (1.45,0.8);
    \node at (2,0) {$\theta_2$};
    \draw (1.9,0.2) -- (1.55,0.8);
    \node at (3.5,0) {$\theta_n$};
    \draw (3.4,0.2) -- (1.65,0.8);
    \node[circle,fill=white,inner sep=4pt] at (2.75,0.5) {$\cdots$};
    \end{tikzpicture}
    \caption{}\label{fig:2complete}
    \end{subfigure}
    \begin{subfigure}{0.45\textwidth}\centering
    \begin{tikzpicture}
    \node at (0,0) {$s^\ast$};
    \node at (0,1) {$m^\ast$};
    \draw[dashed] (0,0.2) -- (0,0.8);
    \node at (1,0) {$\theta_1$};
    \node at (1,1) {$m_{1}$};
    \draw (1,0.2) -- (1,0.8);
    \draw (0.1,0.2) -- (0.8,0.8);
    \node at (2,0) {$\theta_2$};
    \node at (2,1) {$m_2$};
    \draw (2,0.2) -- (2,0.8);
    \draw (0.2,0.2) -- (1.8,0.8);
    \draw (0.3,0.2) -- (3.3,0.8);
    \node[circle,fill=white,inner sep=4pt] at (2.75,0.5) {$\cdots$};
    \node at (3.5,0) {$\theta_n$};
    \node at (3.5,1) {$m_{n}$};
    \draw (3.5,0.2) -- (3.5,0.8);
    \end{tikzpicture}
    \caption{}\label{fig:2star}\end{subfigure}
    \caption{The two types of communication structures stabilized by the lying-averse modification $v_\text{LA}$ defined in eq. \ref{eq:LA}. (a) demonstrates the `single-disclosure' graph, where $m_\emptyset$ may be any non-$\overline{m}$ message, (b) demonstrates the `partial-revelation' graph. The dashed edges may or may not feature.}
    \label{fig:LAgraph}
\end{figure}
\begin{prop}\label{prop:LA}
    Let $\sigma\in \Sigma(u^0)$ be a tree candidate equilibrium of a cheap talk model (identified up to relabelling of messages). If the receiver has unique best responses to pure beliefs, then the following are equivalent:
\begin{enumerate}
    \item $v_\text{LA}\in \SLIDS$,
    \item $v_\text{LA}\in \GCM$,
    \item $G(\sigma)$ is one of the graphs illustrated in Figure \ref{fig:LAgraph}.
\end{enumerate} 
\end{prop}

These graphs are determined by the \emph{half-squares} of their communication graphs.\footnote{\label{fn:halfsquare}The half square $G^2[X]$ of a bipartite graph $G$ contains all the vertices of one side $X$ of the bipartition, and draws edges between vertices that share a neighbour in $G$.

Single-disclosure captures communication graphs whose half-squares $G^2[\Theta]$ and $G^2[M]$ are both complete.

Partial revelation describes communication graphs whose state half-square $G^2[\Theta]$ is a star and whose message half-square $G^2[M]$ is complete.} 
When a communication graph's geometry is complex it becomes impossible to incentivize a specific lie from one state without tempting other states to make the same lie. As a result, there can only be one lying message (a `single-disclosure') as in Figure \ref{fig:2complete}, or there can only be one state that lies (`garbling' the revelation of other messages) as in Figure \ref{fig:2star}.

We use this result to explore the role of money burning in settings with translucent preferences.



\subsubsection*{Money-Burning}
Money burning is well-known technique in communication games to relax incentive constraints. A sender may profit significantly from the receiver taking an action $a$, but this profit may make it impossible to credibly recommend $a$. By accompanying this message with a public burning of a portion of their future profits, this recommendation becomes credible, as this may reduce the incentive to mislead the receiver to take this action.

A key question is whether this process actually benefits the sender, or if it requires the sender to burn all of their profit, providing no benefit over other equilibria.

While money burning greatly expands the set of candidate equilibria in settings with transparent preferences (permitting persuasion for generic priors in non-trivial models), it requires burning \emph{all} of the sender's profit from the induced action, and cannot increase the sender's utility above that which is attainable with cheap talk.\footnote{For discussion/formal proofs of these statements, see Appendix \ref{app:burn}} 

This changes when we move to \emph{nearly} transparent preferences for two reasons. Firstly, adding state-dependence increases the profit that senders obtain from persuasion. In the case of empathy for example, the sender receives a small bonus if communication increases the receiver utility. Thus in models where there is no cheap talk candidate equilibrium, burning money may increase the sender's utility.\footnote{This commonly occurs when the sender has a preference for `higher' beliefs in the receiver (relative to some order on $S$), e.g. the motivating example of the Prosecutor persuading a Judge in \cite{KG11}.} However, this effect factors through the perturbation term, so it will result in only a proportionately small increase in  utility.

A larger effect may occur when cheap talk candidate equilibria require garbling posteriors in a way that is incompatible with the state-dependent preference, while money-burning equilibria persist by permitting different communication graphs. We illustrate how this may occur using a lying aversion modification in a model that extends Example \ref{ex:asset} to managing a portfolio:

\begin{ex}[Advice on an Investment Portfolio]\label{ex:portfolio}
A lying averse broker advises an investor on how to manage their portfolio of investments. The broker possesses information about the state $s$ indicating an optimal balance across assets, however their material compensation is determined by the total value of assets that the investor trades. 

The investor's action now corresponds to a portfolio choice $a\in\mathcal{A}=\Delta^{n-1}$,\footnote{This pure action space, like in Example \ref{ex:asset}, is not discrete. It is easy to extend our results to this model, however the uneasy reader may also consider discrete approximations of these sets.} representing a division of their wealth across $n$ assets. The investor's initial portfolio is $a_0$. Each state $\theta\in \Theta$ corresponds to an optimal portfolio $a^\ast(\theta)$ that the investor seeks to match. As before, the broker receives a fee proportional to the volume of trade from the investors initial portfolio $a_0$. 

We will consider separately the case where the broker can only communicate through cheap talk, and when they may communicate through money burning. In general, the message space will be $M_1\times M_2$ where $M_1=\Theta$ is the set of cheap talk statements and $M_2\in\{\{0\},\mathbb{R}_+\}$ is a set of quantities of money to burn.

The investor and broker material utilities are thus
\begin{align*}
    u_R(a\rvert \theta)=&-\tfrac{1}{2}\|a-a^\ast(\theta)\|_2^2-\kappa\|a-a_0\|_1&
    u^0(a,(m_1,m_2))=&\|a-a_0\|_1-m_2
\end{align*}
where $\kappa\ge 0$ (the reader may consider $\kappa=0$ for simplicity). The broker will also suffer a slight moral cost to lying depending on their cheap talk statement $m_1$, as described by $v_\text{LA}$ (see eq. \ref{eq:LA}). We assume that the investor's portfolio is initially calibrated to their prior belief: $a_0 = \E[\mu]{a^\ast(\theta)}$.

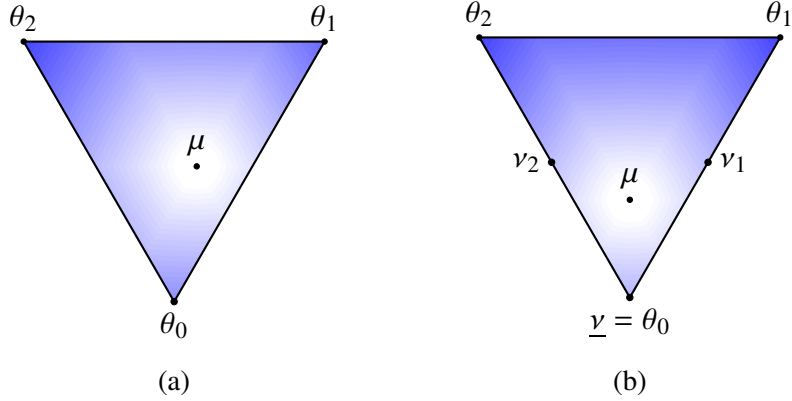
\begin{figure}\centering
\begin{subfigure}{0.3\textwidth}
\centering
\begin{tikzpicture}
    \coordinate (a) at (0,3.46);
    \coordinate (b) at (4,3.46);
    \coordinate (c) at (2,0);
    \newcommand\mux{2.3}
    \newcommand\muy{1.8}
    \newcommand\diff{-0.03}
    \newcommand\magn{25}
    \coordinate (mu) at (\mux,\muy);
    \begin{scope}
    \path[clip] (a) -- (b) -- (c) -- cycle;
    \foreach \x [evaluate=\x as \xn using {\x*2-8}] in {45,...,4} {
    \fill[blue!\xn!white]
    (\mux-2*\x/\magn,\muy) -- 
    (\mux-3/2*\x/\magn+0.865*\diff,\muy+\x*0.865/\magn-1/2*\diff) --
    (\mux-\x/\magn,\muy+\x*1.73/\magn) --
    (\mux,\muy+\x*1.73/\magn-\diff) --
    (\mux+\x/\magn,\muy+\x*1.73/\magn) --
    (\mux+3/2*\x/\magn-0.865*\diff,\muy+\x*0.865/\magn-1/2*\diff) --
    (\mux+2*\x/\magn,\muy) -- 
    (\mux+3/2*\x/\magn-0.865*\diff,\muy-\x*0.865/\magn+1/2*\diff) --
    (\mux+\x/\magn,\muy-\x*1.73/\magn) --
    (\mux,\muy-\x*1.73/\magn+\diff) --
    (\mux-\x/\magn,\muy-\x*1.73/\magn) --
    (\mux-3/2*\x/\magn+0.865*\diff,\muy-\x*0.865/\magn+1/2*\diff) --
    cycle;
    }
    \end{scope}
    \filldraw (a) circle (1 pt) node[anchor = south] {\color{black}$\theta_2$};
    \filldraw (b) circle (1 pt) node[anchor = south] {\color{black}$\theta_1$};
    \filldraw (c) circle (1.2 pt) node[anchor = north] {\color{black}$\theta_0$};
    \draw[thick] (a) -- (b) -- (c) -- cycle;
    \filldraw (mu) circle (1 pt) node[anchor = south] {\color{black}$\mu$};
    \end{tikzpicture}
    \caption{}\label{fig:portfolio.u}
\end{subfigure}\hspace{1cm}
\begin{subfigure}{0.3\textwidth}
\centering
\begin{tikzpicture}
    \coordinate (a) at (0,3.46);
    \coordinate (b) at (4,3.46);
    \coordinate (c) at (2,0);
    \newcommand\mux{2}
    \newcommand\muy{1.3}
    \newcommand\diff{-0.03}
    \newcommand\magn{25}
    \newcommand\ratio{1.73}
    \coordinate (mu) at (\mux,\muy);
    \coordinate (nub) at (\mux+\muy*0.8,\muy*0.8*\ratio);
    \coordinate (nua) at (\mux-\muy*0.8,\muy*0.8*\ratio);
    \begin{scope}
    \path[clip] (a) -- (b) -- (c) -- cycle;
    \foreach \x [evaluate=\x as \xn using {\x*2-8}] in {45,...,4} {
    \fill[blue!\xn!white]
    (\mux-2*\x/\magn,\muy) -- 
    (\mux-3/2*\x/\magn+0.865*\diff,\muy+\x*0.865/\magn-1/2*\diff) --
    (\mux-\x/\magn,\muy+\x*1.73/\magn) --
    (\mux,\muy+\x*1.73/\magn-\diff) --
    (\mux+\x/\magn,\muy+\x*1.73/\magn) --
    (\mux+3/2*\x/\magn-0.865*\diff,\muy+\x*0.865/\magn-1/2*\diff) --
    (\mux+2*\x/\magn,\muy) -- 
    (\mux+3/2*\x/\magn-0.865*\diff,\muy-\x*0.865/\magn+1/2*\diff) --
    (\mux+\x/\magn,\muy-\x*1.73/\magn) --
    (\mux,\muy-\x*1.73/\magn+\diff) --
    (\mux-\x/\magn,\muy-\x*1.73/\magn) --
    (\mux-3/2*\x/\magn+0.865*\diff,\muy-\x*0.865/\magn+1/2*\diff) --
    cycle;
    }
    \end{scope}
    \filldraw (a) circle (1 pt) node[anchor = south] {\color{black}$\theta_2$};
    \filldraw (b) circle (1 pt) node[anchor = south] {\color{black}$\theta_1$};
    \filldraw (nua) circle (1.2 pt) node[anchor = east] {\color{black}$\nu_2$};
    \filldraw (nub) circle (1.2 pt) node[anchor = west] {\color{black}$\nu_1$};
    \filldraw (c) circle (1.2 pt) node[anchor = north] {\color{black}$\underline{\nu}=\theta_0$};
    \draw[thick] (a) -- (b) -- (c) -- cycle;
    \filldraw (mu) circle (1 pt) node[anchor = south] {\color{black}$\mu$};
    \end{tikzpicture}
    \caption{}\label{fig:portfolio.star}
\end{subfigure}
    \caption{An illustration of the sender's indirect utility in Example \ref{ex:portfolio} with $|\Theta|=3$, and each state has its own distinct optimal asset which the investor would like to invest their entire portofio in. Level sets are (slightly rounded) hexagons, due to the hexagonal shape of $\|\cdot\|_1$ on the simplex. Figure \ref{fig:portfolio.star} shows how partial-revelation communication graphs may be persuasive.}
    \label{fig:portfolio}
\end{figure}

Let $\tilde{u}^\ast(\nu):=u^\ast(\nu,0)$ be the broker's indirect utility of a posterior $\nu$ if no money is burnt.  An example case of $\tilde{u}^\ast$ when $|\Theta|=3$ is graphed in Figure \ref{fig:portfolio.u}. For the following proposition, the relevant properties of this model are: (1) $\tilde{u}^\ast$ is quasiconvex; (2) $\tilde{u}^\ast$ attains its minimum at the prior $\mu$.

Let $\theta_0$ be the state that induces the minimum level of trade when revealed: $\theta_0:=\argmin_{\theta}\tilde{u}^\ast(\theta)$. We denote $\mu_{\Theta'}$ to be the distribution that results from conditioning the prior on the event $\Theta'\subseteq \Theta$. A lying averse sender is constrained to the communication graphs of Figure \ref{fig:LAgraph}, resulting in the following proposition.
\begin{prop}\label{prop:portfolio}
    There exists a single-disclosure candidate equilibrium $\sigma\in\Sigma(u^0)$ where the broker attains a material utility of at least $\tilde{u}^\ast(\theta)$ through
    \begin{align*}
    &\text{(a)}\quad\text{cheap talk}& \text{iff}&&
    \tilde{u}^\ast(\theta)<&\tilde{u}^\ast(\overline{\theta})< \tilde{u}^\ast(\mu_{\{\overline{\theta}\}^c})&
    \text{ for some }\overline{\theta},\\
    &\text{(b)}\quad\text{money burning}& \text{iff}&&
    \tilde{u}^\ast(\theta)<& \tilde{u}^\ast(\overline{\theta}),\tilde{u}^\ast(\mu_{\{\overline{\theta}\}^c})&
    \text{ for some }\overline{\theta}.
    \end{align*}
    If $\theta\neq \theta_0$ and $\tilde{u}^\ast(\theta)\ge \tilde{u}^\ast(\mu_{\{\overline{\theta},\theta_0\}})$ for all $\overline{\theta}\in \Theta$ 
    then there exists no partial-revelation candidate equilibrium attaining this material utility.\footnote{This proposition requires $|\Theta|\ge 4$ to be interesting. This is because ruling out the highest state results in a utility of at most the second highest state, so it must be that there are at least two states above $\theta$ for (a,b) to apply. And if there are three states, then (a,b) just show how to attain $\tilde{u}^\ast(\theta_0)$, which is often attainable through partial-revelation graphs (see Figure \ref{fig:portfolio.star}).}
\end{prop}

To attain a utility higher than $\tilde{u}^\ast(\theta)$ with single-disclosure communication graphs, both cheap talk and money burning require that this utility is \emph{ex post} guaranteed by a Blackwell experiment that reveals only whether or not the state is $\overline{\theta}$. However, in cheap talk we additionally requires that the signal $\{\overline{\theta}\}^c$, revealing the state is not $\overline{\theta}$, induces higher utility than $\tilde{u}^\ast(\overline{\theta})$. This is because, in cheap talk, only the former signal can reduce its payoff to provide incentive-compatibility (through garbling the signal) while retaining a single-disclosure graph.

Money-burning avoids this problem by allowing the broker to reduce their payoff from the disclosure signal $\overline{\theta}$ by burning an appropriate amount of money, rather than through garbling the signal.

The last part of the proposition gives conditions under which this is the \emph{only} way to attain such a utility under lying aversion. This condition is satisfied when it is impossible for the broker to garble the signal induced by $\theta_0$ to attain a high utility. The maximum garbling of this signal occurs when the sender at the center state $\theta^\ast=\overline{\theta}$ only ever sends the same signal, resulting in the material utility $\tilde{u}^\ast(\mu_{\{\overline{\theta},\theta_0\}})$.  

\end{ex}

\subsection{Weak Signalling Models}
Our last application is a modification to the communication technology. This, as with lying aversion, turns our communication game into a signalling game, where the signal cost is only weakly state-dependent.

Consider a weakly verifiable disclosure game\footnote{We reach verifiable disclosure games as $\epsilon\rightarrow\infty$. Note that we consider a model allowing vague disclosure, this can be understood as a variation of the `persuasion game' of \cite{M81}. \cite{BC18} provide a broad analysis of such games.} where the sender discloses `pieces of evidence' $e\in \Theta$ to the receiver, each of which rules out its corresponding state --- however the sender may also fabricate evidence at a low cost. This may be the case where the verification process is unreliable, and lies are rarely caught.

We call this model \emph{Weakly Verified Disclosure}. The corresponding modification is then the sum of the cost of presenting the evidence plus the cost of fabricating any false evidence. For a message $m=\{e_s\}_s\in 2^\Theta$ that claims the state is not in $m$, the modification takes the form
\begin{align}
    v_\text{WD}(m\rvert \theta):=&\begin{cases}
    \tau_m&\theta\not\in m\\
    \ell_{m,\theta}& \theta\in m,
    \end{cases}
\end{align}
where $0\ge \tau_m>\ell_{m,\theta}$ for all $\theta\in m$ and messages $m$.

The flexibility in message content allows this modification to be maximally stabilizing:
\begin{prop}\label{prop:WD}
Let $\sigma\in \Sigma(u^0)$ be a tree candidate equilibrium of a cheap talk model (identified up to relabelling of messages). If the message space is $M= 2^\Theta$ then $v_\text{WD}\in \GCM$. If $\ell_{m,\theta}\equiv \ell_m$ then $v_\text{WD}\in \SLIDS$.
\end{prop}
The key fact in this argument is that acyclicity implies equilibrium beliefs have distinct supports, and thus every belief in such an equilibrium may be associated with a distinct truthful message.

\section{Conclusion}
In this paper we discuss how small, state-dependent modifications to the sender's preference can have a discrete effect on the robust equilibria of one-shot communication games. While we focus on finite states, it seems natural that our model should at least approximate interval equilibria in one-dimensional communication games ($\Theta,\mathcal{A}\subseteq \mathbb{R}$) as we increase the number of states/actions. However it is unlikely that this approach will work in multi-dimensional communication games, where the constraint of acyclic communication graphs seems less easily interpretable. Such environments call for different techniques --- I suggest one possibility guess in Appendix \ref{app:CM}.





To outline the scope of our approach, note that in Bayesian Perfect Equilibria of communication games the receiver's decision is determined through the following causal process:
\begin{enumerate}
    \item The sender observes the state and chooses a message,
    \item The receiver observes the message and forms a belief,
    \item The receiver chooses an action that is a best response to this belief.
\end{enumerate}
In this paper we discuss how introducing small state dependence in to the sender's preference over the first and third steps can create robust equilibria when the sender's preferences are otherwise state-independent. Since both correspond to adding a state-dependent component to the sender's utility, they can be studied within the same framework.

The reader might wonder what would be the effect of adding state-dependence to the second step. One way to interpret this is as a \emph{psychological} modification (in the sense of \cite{GPS89}) where the sender's utility is dependent on the receiver's belief. This might be induced through a variation on lying aversion where the sender feels guilt for inducing a misleading belief (see e.g. \cite{KS19}). Since this too is a modification of message-induced subgame utility, it can be treated within the above framework.

A distinct way to introduce state-dependence in the second step is if messages induce beliefs in a state dependent manner. This naturally occurs if the receiver observes an informative signal that is (conditional on the state) independent from the sender's message. Since this signal is informative, it will affect the receiver's beliefs. But this informativeness also means that the likelihood of various signal realizations --- which induce different beliefs --- varies across states, thereby introducing state-dependence in the sender's incentives. As a concrete example, we might consider a consumer who has heard word-of-mouth testimonies about the products a salesperson is promoting. In this way, incentive compatible communication may be possible depending on the inter-dependence of the sender's and receiver's information structures.

\printbibliography
\appendix

\section{Additional Proofs}\label{app:proofs} \subsection*{Topology of Candidate Equilibria}
We state our alternate result for fixed $u_R$ and generic prior beliefs:
\begin{axiom}[Assumption (R)][Receiver Distinguishes between Best Responses.]\label{ax:R}
If $a,a'\in\mathcal{A}$ are receiver-best responses to a belief $\nu\in\Delta \Theta$ then there exists a state $\theta\in\supp{\nu}$ such that $u_R(a\rvert \theta)\neq u_R(a'\rvert \theta)$.
\end{axiom}
This says that when restricted to any subset $\Theta_1\subseteq \Theta$, if two actions are equivalent to the reciever (ie. they yield identical utilities for all states in $\Theta_1$), they must be irrelevant (ie. they are never best responses for beliefs in $\Theta_1$). As a result the set of beliefs in $\Delta \Theta_1$ where the receiver is indifferent between two best responses has codimension 1.

There is also a slightly stronger version of this assumption that we will also use:
\begin{axiom}[Assumption (R$\ast$)][Receiver Distinguishes within 3-Best Responses.]
Let $a,a',a''\in\mathcal{A}$ be (possibly non-distinct) receiver-best responses to a belief $\nu\in\Delta \Theta$. If $p,p'\in\Delta\{a,a',a''\}$ and $u_R(p\rvert \theta)= u_R(p'\rvert \theta)$ for all $\theta\in\supp{\nu}$, then $p=p'$.
\end{axiom}
This can be equivalently stated in terms of the preference over pure actions as follows: if $a$ is equivalent to a convex combination of $a',a''$ (when restricted to any $\Theta_1\subseteq \Theta$) then $a$ is irrelevant (over $\Delta \Theta_1$).

This implies that for any $\Theta_1\subseteq \Theta$, (i) the set of beliefs in $\Delta \Theta_1$ where the receiver is indifferent between two best responses has codimension 1, and (ii) the set of beliefs in $\Delta \Theta_1$ where the receiver is indifferent between three best responses has codimension 2. 

These assumptions are necessary conditions for an equilibrium involving such mixed actions to be Harsanyi stable to perturbations in the \emph{receiver's} utility. They are generically satisfied on the function space $\mathbb{R}^{\mathcal{A}\times \Theta}$. 

We now extend our observation about pure-strategy equilibria to acyclic equilibria:

\begin{manualtheorem}{1'}\label{thm:genmu}
If assumption (R) is satisfied and the sender has transparent preferences, then for generic priors $\mu$
\begin{enumerate}[label=(\textbf{\alph*})]
    \item there are no forest equilibria,
    \item all acyclic equilibria involve precisely one pure action.
\end{enumerate}
When assumption (R$\ast$) is satisfied, we further obtain that for generic priors $\mu$
\begin{enumerate}[resume,label=(\textbf{\alph*})]
    \item all acyclic equilibria have only binary mixed actions (ie. $p$ with $|\supp{p}|=2$),
    \item there are finitely many acyclic equilibria (up to choice of messages), each uniquely identified by its communication graph.
\end{enumerate}
\end{manualtheorem}
\begin{proof}[Proof of Theorem \ref{thm:gen} and \ref{thm:genmu}]
    Both for  is a corollary of the following lemma:
\begin{lem}\label{lem:forest}
    Under Assumptions (S), for generic $\mu$ every acyclic connected component of every communication graph consists of
    \begin{itemize}
        \item a pure action when Assumption (R) is satisfied, furthermore there are only finite posteriors inducing mixed actions represented by this communication graph
        \item furthermore, mixed actions randomizing between (only) two pure actions when Assumption (R$\ast$) is satisfied.
    \end{itemize}
    Alternatively these properties are satisfied for fixed $\mu$ and generic $u_R$.
\end{lem}
Observe that since pure actions are transparently ranked, there is no equilibrium involving multiple pure actions. Thus the first property implies (a)-(c).
\end{proof}
To prove this lemma, we first establish some notation: 

For a rooted tree $G_0$, for a node $j$, we denote the set of $j$'s children by $j^\downarrow$ and its parent by $j^\uparrow$. This will often be pushed to subscripts --- eg. a message $\pi_j$ will have children from the set of states $\theta_{j^\downarrow}$. The set of $j$'s neighbours will be denoted $\mathcal{N}(j):=\{j^\uparrow\}\cup j^\downarrow$.

We also require notation for the sets of beliefs where the receiver is indifferent between multiple actions. For a set of pure actions $A\subseteq\mathcal{A}$, define the $A$ best-response set 
\begin{align*}
    \mathcal{I}(A):=&\left\{\nu\in \Delta \Theta; 
    A\subseteq \argmax_{a\in \mathcal{A}}\E[s\sim\nu]{u_R(a,s)} 
    \right\}
\end{align*}
be the set of beliefs such that the receiver is willing to randomize over $A$ (ie. choosing an action $p\in\inte{\Delta A}$).

We let 
\begin{align*}
    \mathcal{I}_2:=&\{\mathcal{I}(A);|A|=2\}&
    \mathcal{I}_{3+}:=&\{\mathcal{I}(A);|A|>2\}&
\end{align*}
refer to collections of these sets where the receiver is willing to randomize between 2 actions, and 3 or more actions respectively. Assumption (R) implies every set in $\mathcal{I}_2$ has codimension 1, while (R$\ast$) implies every set in $\mathcal{I}_{3+}$ has codimension at least 2.
\begin{proof}[Proof of Lemma \ref{lem:forest}]

We first prove the statement of the theorem for generic priors, before showing how this proof can be translated to show the alternative statement.

WLOG we limit ourselves to priors $\mu\in\inte{\Delta \Theta}$. Suppose there is an acyclic connected component $G_0\subseteq G$ of the communication graph that does not include a pure action. We turn $G_0$ into a rooted tree with arbitrary message root $\pi_0$. We make an inductive argument up the tree to the root:

\noindent\textbf{1. Leaf Case:}
We begin with the vertices furthest from the root. Since there are no pure actions in the graph, all the leaves are states (otherwise they are actions corresponding to a known state --- thus pure under Assumption (R)). The deepest vertex thus is a state connected to a single message, say $\pi_j$, which it recommends w.p. 1 in equilibrium. Since this is a deepest action, all its other children $\theta_{k}\in \Theta_{j^\downarrow}$ also recommend it w.p. 1. 

The only degree of freedom that $G_0$ admits in the belief $\nu_j$ is the probability that it is recommended from its parent state $\theta_{j^\uparrow}$. Thus for a fixed prior, the set of possible $\nu_j$ corresponding to this graph has dimension 1. 

Fixing only the mass $\mu_{j^\uparrow}\in]0,1[$ of the parent state, the set of possible posteriors $\nu_j$ under $G$ is given by
$$\mathcal{P}_{j}:=\Big\{\lambda \delta_{\theta_{j^\uparrow}}+(1-\lambda)\mu_{j^\downarrow};\,\lambda\in \big]0,1\big[,\mu_{j^\downarrow}\in \inte{\Delta \Theta_{j^\downarrow}}\Big\}\subseteq \Delta \Theta_{\mathcal{N}(j)}.$$ 
Note that this set is an open subset of $\Delta \Theta_{\mathcal{N}(j)}$, hence.has codimension 0. Thus it is trivially transversal to any best response set. 

Applying the transversality theorem, for any fixed $\mu_{j^\downarrow}$, the set
$$
\mathcal{P}_{j}(\mu_{j^\downarrow}):=\Big\{\lambda \delta_{\theta_{j^\uparrow}}+(1-\lambda)\mu_{j^\downarrow};\,\lambda\in \big]0,1\big[\Big\}
$$
is transversal to any $\mathcal{I}(A)$ for a.a. $\mu_{j^\downarrow}\in \Delta \Theta_{j^\downarrow}$. Thus, within this generic set, $\mathcal{P}_{j}(\mu_{j^\downarrow})$ intersects $I\in\mathcal{I}_2$ at finitely many points, and, under assumption (R$\ast$), never intersects the codimension 2 set $I\in\mathcal{I}_{3+}$.

Fixing such a $\mu_{j^\downarrow}$, there are finite probabilities that the sender can send the message $\pi_j$ from state $\theta_{j^\uparrow}$ and induce a mixed action. 

\noindent\textbf{2. Induction step:}
We now let $\pi_j$ be a message further up the tree $G_0$, and $\theta_{j^\uparrow}$ be its parent state. 
We will take as fixed the probability $\mathscr{M}_{j^\downarrow}$ that their child states $\theta_{k}\in \Theta_{j^\downarrow}$ recommend their own child actions $\pi_{k^\downarrow}$ (and thus the posteriors associated with these actions). Our inductive hypothesis will be that, generically there is a finite (possibly empty) set $\hat{m}_{j^\downarrow}$ of $\mathscr{M}_{j^\downarrow}$ that allow the receiver to randomize between actions.

We then claim that for generic prior weights $\mu_{j^\downarrow}\in \Delta \Theta_{j^\downarrow}$ on $\pi_j$'s children, there are finite probabilities $\lambda_j$ that its parent state can send the message $\pi_j$.

Indeed, fixing the posteriors generated in descendent actions, we find that the range of posteriors that can be generated by the message $\pi_j$ is given by
$$
\mathcal{P}_{j}:=\Big\{\lambda \delta_{\theta_{j^\uparrow}}+(1-\lambda)(1-\mathscr{M}_{j^\downarrow})\mu_{j^\downarrow};\, \lambda\in \big]0,1\big[,\mu_{j^\downarrow}\in \inte{\Delta \Theta_{j^\downarrow}}\Big\}\subseteq \Delta \Theta_{\mathcal{N}(j)}.
$$
Note that here, we are writing the prior as a convex combination of $\delta_{\theta_{j^\uparrow}}$ and \emph{all} its descendent actions, which allows us to keep the sender strategy in descendent states, in particular $\mathscr{M}_{j^\downarrow}$, fixed.  This set has codimension 0. 

Applying the transversality theorem once more, for any fixed $\mu_{j^\uparrow}$, the set of feasible posteriors
$$
\mathcal{P}_{j}(\mu_{j^\downarrow}):=\Big\{\lambda \delta_{\theta_{j^\uparrow}}+(1-\lambda)(1-\mathscr{M}_{j^\downarrow})\mu_{j^\downarrow};\, \lambda\in \big]0,1\big[\Big\}\subseteq \Delta \Theta_{\mathcal{N}(j)}
$$
is thus transversal to any $\mathcal{I}(A)$ for a.a. $\mu_{j^\downarrow}\in \Delta \Theta_{j^\downarrow}$. Again, as a result, $\mathcal{P}_{j}(\mu_{j^\downarrow})$ intersects $I\in\mathcal{I}_2$ at finitely many points, and (under assumption (R$\ast$)) never intersects any $I\in\mathcal{I}_{3+}$ for such $\mu_{j^\downarrow}$.

Letting $\mathscr{M}_{j^\downarrow}$ vary over the finite set $\hat{m}_{j^\downarrow}$, we get the finite intersection property for \emph{any} sender strategy corresponding to $G_0$ whenever $\mu_{j^\downarrow}$ lies within a finite intersection of generic sets.


\noindent\textbf{3. Conclusion:} This shows that generically (non-root) actions have support of at most two. But if we consider the rooted message, the above reasoning shows that there are finite posteriors that can be generated when all its descendants are constrained to generate mixed actions. Applying the transversality theorem (as above, but with no parent distribution $\mu_{j^\uparrow}$) we find that for a.a. $\mu_{j^\downarrow}$ the posteriors do not intersect any set $I\in\mathcal{I}_2$ of posteriors inducing a mixed action. Thus one action must be pure. 

Since there are finite acyclic communication structures, the set of priors that admit a forest candidate equilibrium is the finite union of measure-0 sets, and hence is measure-0.

\textbf{Translation to proof of alternate statement: }The alternate statement is that the above results also hold for fixed $\mu$ and generic $u_R$. We make Assumption (R$\ast$) (recalling that it holds for generic receiver preferences), and note that the above theorem can be reproduced by holding $\mu_{j^\downarrow}$ fixed at each step, and varying $I(A)\in\mathcal{I}_2$ through the receivers utility.

    Observe that the set $\{I(A;u_R)\}$ can be perturbed in any direction by adjusting $u_R$ within this set of utilities, hence it is trivially transversal to $\mathcal{P}_j(\mu_+)$. Thus, for generic $u_R$, $I(A; u_R)$ is transversal to $\mathcal{P}_j(\mu_+)$. Since the latter is one dimensional (parametrized solely by the weight its parent state puts on it), and the former has codimension one (by Assumption (R$\ast$)), this means that the intersection consists of a finite set of weights the parent state can send the message with that will induce randomization over two actions (and no weights that will induce higher support randomizations).
\end{proof}

\subsection*{Fragility Proofs}
\begin{proof}[Proof of Proposition \ref{prop:Frag}]
\textbf{Proof of (1): } Suppose there is a persuasive equilibrium $\sigma$. Let $\mathscr{P}(M)\subseteq \Delta \mathcal{A}$ be the set of mixed actions that are chosen after some message. Since $\sigma$ is persuasive, there is at least one pure action $a^\ast$ that is not best response to the prior, yet occurs with positive probability --- ie. if $M^\ast:=\mathscr{P}^{-1}\{p;a^\ast\in \supp{p}\}$ is the set of messages leading to $a^\ast$ being played with positive probability, then the sender sends messages in $M^\ast$ with positive probability, ie. $\P{\mathscr{M}(\theta,\omega)\in M^\ast}>0$.

Let $\Pi^\ast:=\mathscr{P}(M^\ast)\times M^\ast$ be the set of equilibrium paths leading to action $a^\ast$, and $\Pi^c$ be the remaining equilibrium paths. We obtain the following bound:
\begin{equation*}
    \P{\max_{\pi\in \Pi^\ast} u_S(\pi)> \max_{\pi'\in \Pi^c} u_S(\pi')}\le \Pcond{\mathscr{M}(\theta,\omega)\in M^\ast}{\theta}\le \P{\max_{\pi\in \Pi^\ast} u_S(\pi)\ge \max_{\pi'\in \Pi^c} u_S(\pi')}.
\end{equation*}
However
\begin{equation*}
    \P{\max_{\pi\in \Pi^\ast} u_S(\pi)=\max_{\pi'\in \Pi^c} u_S(\pi')}=\E{\Pcond{\max_{\pi\in \Pi^\ast} u_S(\pi)=\max_{\pi'\in \Pi^c} u_S(\pi')}{\omega_1,(u_S(a))_{a\neq a^\ast}}}.
\end{equation*}
Since $u_S(\pi')$ is determined by $(\omega_1,(u_S(a))_{a\neq a^\ast})$, the conditional probability amounts to the probability that $\max_{\pi\in \Pi^\ast} u_S(\pi)$ is equal to a constant. But $\max_{\pi\in \Pi^\ast} u_S(\pi)$ is strictly increasing in $U(a^\ast)$, so there is a unique value of $U(a^\ast)$ that will yield this equality. Since the utility admits a density, the conditional probability of $U(a^\ast)$ being precisely this value is $0$, allowing us to conclude 
\begin{equation*}
    \Pcond{\mathscr{M}(\theta,\omega)\in M^\ast}{\theta}=\P{\max_{\pi\in \Pi^\ast} u_S(\pi)> \max_{\pi'\in \Pi^c} u_S(\pi')}
\end{equation*}
is state-independent. As $\P{\mathscr{M}(\theta,\omega)\in M^\ast}>0$, this means the average posterior that induces a randomization including $a^\ast$ is equal to the prior $\mu$. Thus the prior is a convex combination of beliefs to which $a^\ast$ is a best response, indicating $a^\ast$ is a best response to the prior, contradicting our premise.

\textbf{Proof of (2): } Meagreness is a corollary of (1), under the following two observations:
\begin{enumerate}
    \item The set of preferences described by eq. \ref{eq:frag} is dense.
    \item The set of preferences that permit persuasive equilibria is the countable union of closed sets.
\end{enumerate}
The first property is trivial. To observe the second, consider the set of persuasive equilibria where an action $a^\ast$ that is not a best response to the prior is chosen w.p. at least $\frac{1}{n}$. This is a closed set of equilibria, by upper-hemicontinuity the corresponding set of preferences is also closed. We know that the complement is then open, and by the first property includes a dense set, thus these equilibria only occur over a nowhere dense set of preferences.

By taking the union of these preferences over $n$, we find that persuasion is only possible within a meagre set (ie. a countable union of nowhere dense sets).
\end{proof}

Note that the topological properties invoked in the last proof are desirable for any topology on random preferences. Thus persuasion is generically impossible for any `reasonable' topology, not just the weak-$\ast$ topology.
\subsection*{$\mathscr{O}$-stability Proofs}
\begin{proof}[Proof of Lemma \ref{lem:NIC}]
Fix $B_0:=\{\pi_0\}$. Given actions $\pi\in B:=\bigtimes B_j$, the sender will send the message $\pi_j$ with probability $\hat{m}_j(\pi)$ satisfying
\begin{equation*}
    \P[\omega]{u_S(\pi_j)=\max_{j'}u_S(\pi_{j'})}\ge \hat{m}_j(\pi)\ge \P[\omega]{u_S(\pi_j)>\max_{j\neq j'}u_S(\pi_j)}.
\end{equation*}
When idiosyncrasy admits a density, $\hat{m}$ is a singleton. More generally $\hat{m}(\pi)$ is the intersection of the `box' in $[0,1]^\mathcal{N}$ whose sides described by the above interval, with the simplex $\Delta\mathcal{N}$. Let 
\begin{align*}
    \underline{\pi}_j:=&\argmin\{\P[\omega]{u_S(\pi')>u_S(\pi_0)}\}&
    \overline{\pi}_j:=&\argmax\{\P[\omega]{u_S(\pi')>u_S(\pi_0)}\}
\end{align*}
By eq. \ref{eq:NIC.bound}, the action $\underline{\pi}_j$ makes it a best response to never send its message: $0\in\hat{m}_j(\underline{\pi}_j,\pi_{-j})$ for any $\pi_{-j}\in B_{-j}:=\bigtimes_{j'\neq j}B_{j'}$; while a receiver having the strategy corresponding to $\overline{\pi}_j$ ensures that $\pi_0$ will never be sent in equilibrium, ie. $\{0\}=\hat{m}_0(\overline{\pi}_j,\pi_{-j})$. Note that ordering $B_j$ so that $\overline{\pi}_j\succ \underline{\pi}_j$ means the senders' preference will be $\succ$-monotone (as a consequence of vNM), implying $\hat{m}_j(\cdot,\pi_{-j})$ is $\succ$-increasing.

Define the correspondence $g^\ast_j:B_{-j}\Rightarrow B_j$
\begin{equation*}
    g^\ast_j(\pi'_{-j}):=\begin{cases}
    \{\overline{\pi}_j\}& \hat{m}_j (\overline{\pi}_j;\pi'_{-j})<m_j^\ast\\
    \{\pi'_j\in B_j;\hat{m}_j (\pi'_j;\pi'_{-j})=m_j^\ast\}&\text{otherwise}
    \end{cases}
\end{equation*}
Note this is never empty-valued --- since $\hat{m}_j(\cdot,\pi'_{-j})$ is increasing and upper hemicontinuous it satisfies an intermediate value theorem and eq. \ref{eq:NIC.bound} implies $\hat{m}_j (\underline{\pi}_j;\pi'_{-j})< m_j^\ast$ for any $\pi'_{-j}$. This further shows $g^\ast_j$ is upper-hemicontinuous and interval-valued.


Define the self-map $g:B\Rightarrow B$
\begin{equation*}
    g(\pi')=\bigtimes_{j}g^\ast_j(\pi'_{-j})
\end{equation*}
By Kakutani fixed point theorem, there exists a fixed point $\pi^\ast$ of this map. If this fixed point does not solve $\hat{m}_j (\pi^\ast)=m_j^\ast$, then $\pi_{j}^\ast=\overline{\pi}_j$ for at least one $j$. But then $\hat{m}_0(\pi^\ast)<m_0^\ast$, as earlier observed. As a result there must be a message $j'$ (not necessarily $j$) sent with probability $\hat{m}_{j'} (g^\ast_{j'}(\pi^\ast_{-j'});\pi^\ast_{-j'})>m_{j'}^\ast $, contradicting the definition of $g_{j'}^\ast$.
\end{proof}
\begin{proof}[Proof of Theorem \ref{thm:mainB}]
We use the same rooted tree as in the proof of Theorem \ref{thm:main}, with the same notation.

\textbf{1. The Neighbour Problem:} To ensure each state on the tree randomizes over its neighbouring messages to the desired degree, first constrain each state to neighbouring messages. Suppose $B_j\ni \pi_j$ are interval neighbourhoods of $\pi_j$ for all non-pure $\pi_j$ (let $B_0:=\{\pi_0\}$ for the pure action). We assume $\epsilon$ is sufficiently small that the sender's ordinal preference over pure actions in independent of idiosyncrasy. Thus we can write each $B_j=:[\underline{\pi}_j,\overline{\pi}_j]$ where actions on the right side of the interval are unanimously prefered to actions to their left.

We seek $\epsilon>0$ and interval neighbourhoods $B_j$ such that
\begin{align}
    u_S(\underline{\pi}_{k^\downarrow}\rvert \theta_k)<u_S(\pi_{k^\uparrow}\rvert \theta_k)<u_S(\overline\pi_{k^\downarrow}\rvert \theta_k)\qquad\text{w.p. 1, for all $\pi_{k^\uparrow}\in B_{k^\uparrow}$, $k$}.
\end{align}
This allows us to apply Lemma \ref{lem:NIC} to find mixed actions $\hat{\pi}\in B:=\bigtimes B_j$ such that each state sends its neighbouring messages with the desired probability.

We take a fixed constraint that our equilibrium actions must lie in neighbourhood $N_j$ of $ \pi_j$, and specify the constraints this imposes on $\epsilon$ and grandparent actions $\pi_{j'}$. We begin at the bottom of the tree and work our way up to the root.

For mixed actions $\pi_j$ without grandchildren, define $B_j:=N_j$. For mixed actions with grandchildren, define
$$
    B_{k^\uparrow}:=\bigcap_{j'\in k^\downarrow}\{\pi'\in N_{k^\uparrow}; u_S(\underline\pi_{j'}\rvert \theta_k) \le u_S(\pi'\rvert \theta_k)\le u_S(\overline\pi_{j'}\rvert \theta_k)\text{ w.p. 1}\}.
$$
This may be empty for some $\epsilon$, but observe that $B_{k^\uparrow}\rightarrow O_{k^\uparrow}$ where
$$
 O_{k^\uparrow}:=\left\{\pi'\in N_{k^\uparrow};u^0(\pi')\in\Big]\max_{j'\in k^\downarrow}\min u^0(B_{j'}),
    \min_{j'\in k^\downarrow}\max u^0(B_{j'})\Big[\right\}\ni \pi_{k^\uparrow}.
$$
whenever $B_{j'}\ni \pi_{j'}$. Thus when $\epsilon$ is sufficiently small, $B_{k^\uparrow}$ is a neighbourhood of $\pi_{k^\uparrow}$. We can then work up the tree until be simultaneously have neighbourhoods around every mixed action, and can apply Lemma \ref{lem:NIC} to get neighbouring incentive compatibility for every state.

\textbf{2.1 Harsanyi stability} To get the limit (1) in the second part of the theorem, bound $\epsilon<\overline{\epsilon}$ so that the proof of Theorem \ref{thm:main} can be applied to ensure states have strict preferences for their own limit equilbirium actions at $v$. Then repeat the above step with $V\drightarrow v$ instead of $\epsilon\rightarrow 0$, and $\pi$ representing the equilibrium actions of $u^0+\epsilon v$ rather than $u^0$.



\medskip

\noindent\textbf{3. Non-neighbouring deviations} Assume that the neighbourhoods $N_j$ collectively solve eq. \ref{eq:GCM}. Once we have actions $\hat{\pi}$ that induce the desired degree of randomization when constrained to on-path messages, we need to check that agents are not tempted by off-path messages. To observe a sender with realized state-idiosyncrasy $(\theta_1,\omega)$ will have strict preferences for neighbouring messages, consider a sequence $(\theta_1,\dots,\pi_N)$. Applying intermediate value theorem, we find a modification $v^\ast\in \text{co}(\supps{u_S})\subseteq \GCM$ that is indifferent over each message is this sequence. Simply apply the equivalent part of the proof of Theorem \ref{thm:main} to the modification
$$
\hat{v}(\pi\rvert \theta):=\begin{cases}
    v^\ast(\pi\rvert \theta)&\theta\neq \theta_1\\
    V(\pi\rvert \theta_1,\omega)& \theta=\theta_1.
\end{cases}
$$
to show that the sender with state/idiosyncrasy $(\theta_1,\omega)$ will not want to deviate to non-neighbouring actions.
\end{proof}

\subsection*{Ordinal Transparence Proofs}
\begin{proof}[Proof of Proposition \ref{prop:translucent}]
    Suppose $\sigma$ is an equilibrium of a model with ordinally transparent preference $u_S$, and let $u^0$ be a transparent representation of $u_S$.

    Note that $\sigma$ is informationally equivalent to a candidate equilibrium $\sigma'$ of the transparent preference $u^0$ if the messages it sends in equilibrium $\Pi\subseteq \Delta\mathcal{A}\times M$ satisfy
\begin{align*}
        \bigcap_{\pi\in\Pi}&[\min u^0(\supp{\pi}),\max u^0(\supp{\pi})]\neq \emptyset,\\
        \min_{\pi\in\Pi}&\max u^0(\supp{\pi})\ge \max_{m\in M}\min_{\nu\in\Delta \Theta}\{u^\ast(\nu,m)\}.
\end{align*}
where the second inequality comes from increasing all the equilibrium actions (while maintaining indifference) until one of them is pure, and checking this satisfies eq. \ref{eq:offpath}.

If the first inequality fails then $\sigma$ must involve actions $\pi,\pi'$ such that $\max u^0(\supp{\pi})< \min u^0(\supp{\pi'})$. But then ordinal transparency implies that $u_S(\pi)<u_S(\pi')$ w.p. 1 and hence $\pi\not\in\Pi$.

The second inequality holds since whichever senders $(s,\omega)$ send the message that minimizes the left side must prefer it (and therefore the preferred pure action in its support) to the pure action on the righthand side.
\end{proof}
\begin{proof}[Proof of Proposition \ref{prop:sandwich}]
    We proceed by induction on $N$. When $N=1$, graph-cyclic monotonicity is equivalent to strong graph-monotonocity. 
    
    Now suppose SGM and GM hold on paths of length $2N$, we want to show that
    \begin{equation*}
            \sum_{i=0}^{N} v_i(\pi_{i})-v_i(\pi_{i-1})>0
    \end{equation*}
    for paths $(\theta_0,\pi_0,\dots,\theta_N,\pi_N=:\pi_{-1})$ of length $2N+2$. Our induction hypothesis applies to show that GM holds for the path $(\theta_1,\pi_1,\dots,\theta_N,\pi_N)$, ie.
    \begin{equation*}
            v_1(\pi_1)-v_1(\pi_{N})+\sum_{i=2}^{N} v_i(\pi_{i})-v_i(\pi_{i-1})>0.
    \end{equation*}
    Adding this to the SGM inequality
    \begin{equation*}
        v_0(\pi_0)-v_0(\pi_N)-(v_1(\pi_0)-v_1(\pi_N))\ge 0
    \end{equation*}
    immediately yields our conclusion.
\end{proof}
\subsection*{Proofs for Applications}
\begin{proof}[Proof of Proposition \ref{prop:emp}]
\noindent\textbf{Antipathy case: }Consider a persuasive candidate equilibrium $\sigma$. If $\Econd{u_R(p_j\rvert \theta)}{m_i}=\max_{a\in\mathcal{A}}\Econd{u_R(a\rvert \theta)}{m_i}$ for all messages $m_i$ and equilibrium actions $p_j$ then the equilibrium is not persuasive, as any equilibrium action is optimal for any message, hence under the prior.

Thus let $p_0$ be the equilibrium action after some message $m_0$, and $p_1$ the action after a distinct message $m_1$ such that
\begin{align*}
    \Econd{u_R(p_0\rvert \theta)-u_R(p_1\rvert \theta)}{m_0}>& 0\ge \Econd{u_R(p_0\rvert \theta)-u_R(p_1\rvert \theta)}{m_1}.
\end{align*}
As a result, for at least one state $\theta_0\in \mathscr{M}^{-1}(m_0)$ and $\theta_1\in \mathscr{M}^{-1}(m_1)$, we have
\begin{align*}
    u_R(p_0\rvert \theta_0)-u_R(p_1\rvert \theta_0)>&0&
    u_R(p_1\rvert \theta_1)-u_R(p_0\rvert \theta_1)\ge&0.
\end{align*}
This means that the sender preference differentials
\begin{align*}
    u_S(p_0\rvert \theta_1)-u_S(p_0\rvert \theta_0)=&\epsilon(u_R(p_0\rvert \theta_0)-u_R(p_0\rvert \theta_1))\\
    u_S(p_1\rvert \theta_0)-u_S(p_1\rvert \theta_1)=&\epsilon(u_R(p_1\rvert \theta_1)-u_R(p_1\rvert \theta_0))
\end{align*}
sum to a strictly positive number. But sender incentive compatibility requires that both of these differentials be weakly negative.

\noindent\textbf{Empathy case: } We will show that for any path $(\theta_0,\pi_0,\dots,\theta_N,\pi_N)$ where $(\theta_0,\dots,\theta_N)$ is $\succ$-monotone (WLOG increasing), we have
\begin{align*}
     u_R(\alpha_0\rvert \theta_0)-u_R(\alpha_N\rvert \theta_0)> u_R(\alpha_0\rvert \theta_1)-u_R(\alpha_N\rvert \theta_1)
\end{align*}
for all $\alpha_i\in\supp{\pi_i}$. Single-crossing implies that for a fixed $\alpha_0,\alpha_N$ pair either
\begin{align*}
    u_R(\alpha_0\rvert \theta)-u_R(\alpha_N\rvert \theta)> u_R(\alpha_0\rvert \theta')-u_R(\alpha_N\rvert \theta') &\quad\text{for all }\theta'\succ \theta,\text{ or}\\
    u_R(\alpha_0\rvert \theta)-u_R(\alpha_N\rvert \theta)< u_R(\alpha_0\rvert \theta')-u_R(\alpha_N\rvert \theta') &\quad\text{for all }\theta'\succ \theta.
\end{align*}
Note that our equilibrium requires that $\alpha_N$ is a best response to a belief supported on states $s'\succeq \theta_N$, and that $\alpha_0$ is a best response to a belief supported on states $s\preceq \theta_N$. This is incompatible with the second possibility (unless both beliefs are the same degenerate distribution $\delta_{\theta_N}$), so the first must hold, validating support monotonicity.
\end{proof}

\begin{proof}[Proof of Proposition \ref{prop:antipath}]
In Appendix \ref{app:weakstab}, we observe that in three-action models graph-cyclic monotonicity is necessary to preserving a candidate equilibrium (consequently no cyclic equilibrium exists, and by Theorem \ref{thm:gen} for generic prior it is sufficient to consider connected candidate equilibria).

    Acyclic candidate equilibria that attain a non-zero payoff will involve two messages: one recommending mixed action $p_B$, purchasing good $B$ and not buying anything with equal probability; and one recommending pure action $A$\footnote{A measure 0 of priors will also allow an equilibrium where mixed action $p_A$ is recommended. These equilibria are fragile to perturbations to the \emph{receiver}'s utility. However, they are still acyclic and can be analyzed from the perspective of the sender modifications.}.

We can see that a receiver only chooses action $A$ if $a$ is in the support of their posterior, likewise for $B$ and $b$. Thus these two action-state pairs must be connected on the communication graph. 

With an eye towards graph-cyclic monotonicity, we can see that 
\begin{align*}
    v_E(p_B\rvert b)-v_E(A\rvert b)=\tfrac{7}{8}&>
    v_E(p_B\rvert a)-v_E(A\rvert a)=-\tfrac{3}{8}>v_E(p_B\rvert n)-v_E(A\rvert n)=-\tfrac{5}{8}.
\end{align*}
Thus if state $n$ sends both messages, state $a$ can only send the message recommending action $p_B$, yielding a contradiction as the receiver would never choose $A$ then. If state $b$ recommends both actions, then graph-cyclic monotonicity implies that all other states only recommend $A$. But then the message recommending $p_B$ induces a pure belief $\delta_b$ to which $p_B$ is not a best response.

Thus in any equilibrium $\sigma$ where $u_R\in \GCM$, $a$ recommends both actions, while $b$ only recommends $p_B$ and $n$ only recommends $A$. This gives us the unique connected communication structure $G$ such that empathy satisfies graph-cyclic monotonicity. However, under cheap talk, there are no candidate equilibrium with this communication structure. 

The reader may verify there are acyclic candidate equilibria that are plausible for the prior beliefs in the cross-hatched region, and thus may be preserved by appropriate modifications.
\end{proof}

The following lemma is essential to the proof of Proposition \ref{prop:LA}:
\begin{lem}[N-Shaped Subgraphs]\label{lem:N}
    Consider a lying averse model, and a communication graph containing an \emph{N}-shaped subgraph, in that two states $\theta_1,\theta_2$ both send a message $m_1$ and one of them ($\theta_2$) sends another message $m_2$ in the candidate equilibrium $\sigma$. If $v_\text{LA}\in\GCM$, then either $\psi(m_1)=\theta_1$ or $\psi(m_2)=\theta_2$.
\end{lem}
The N-shaped subgraph is illustrated in Figure \ref{fig:N}. Notationally, we will say such an `N' is \emph{spanned by} its endpoints --- in this case $\theta_1$-$m_2$. The content of this lemma is intuitive: for the modification to be consistent with this `N', the sender in state $\theta_1$ must be biased towards the message they more than the sender in state $\theta_2$. The only way this can happen is if they are telling the truth ($\psi(m_1)=\theta_1$) or the other message is the truth for the sender in state $\theta_2$ ($\psi(m_2)=\theta_2$).  
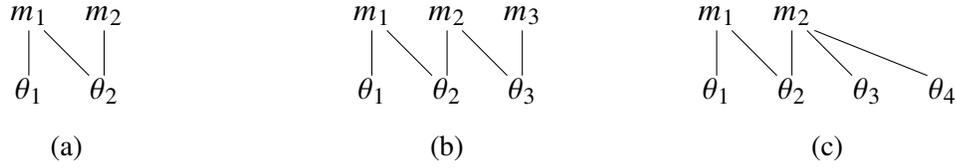
\begin{figure}
    \centering
    \begin{subfigure}{0.3\textwidth}\centering
    \begin{tikzpicture}
    \node at (1,0) {$\theta_2$};
    \node at (0,0) {$\theta_1$};
    \node at (0,1) {$m_1$};
    \node at (1,1) {$m_2$};
    \draw (0,0.2) -- (0,0.8);
    \draw (1,0.2) -- (1,0.8);
    \draw (0.2,0.8) -- (0.8,0.2);
    \end{tikzpicture}
    \caption{}\label{fig:N}
    \end{subfigure}
    \begin{subfigure}{0.3\textwidth}\centering
    \begin{tikzpicture}
    \node at (2,0) {$\theta_3$};
    \node at (1,0) {$\theta_2$};
    \node at (0,0) {$\theta_1$};
    \node at (0,1) {$m_1$};
    \node at (1,1) {$m_2$};
    \node at (2,1) {$m_3$};
    \draw (0,0.2) -- (0,0.8);
    \draw (1,0.2) -- (1,0.8);
    \draw (0.2,0.8) -- (0.8,0.2);
    \draw (2,0.2) -- (2,0.8);
    \draw (1.2,0.8) -- (1.8,0.2);
    \end{tikzpicture}
    \caption{}\label{fig:imp1}
    \end{subfigure}
    \begin{subfigure}{0.3\textwidth}\centering
    \begin{tikzpicture}
    \node at (3,0) {$\theta_4$};
    \node at (2,0) {$\theta_3$};
    \node at (1,0) {$\theta_2$};
    \node at (0,0) {$\theta_1$};
    \node at (0,1) {$m_1$};
    \node at (1,1) {$m_2$};
    \draw (0,0.2) -- (0,0.8);
    \draw (1,0.2) -- (1,0.8);
    \draw (0.2,0.8) -- (0.8,0.2);
    \draw (2.8,0.2) -- (1.3,0.8);
    \draw (1.2,0.8) -- (1.8,0.2);
    \end{tikzpicture}
    \caption{}\label{fig:imp2}
    \end{subfigure}
    \caption{In (a), the communication sub-graph referred to in Lemma \ref{lem:N}. In (b), (c) are the two sub-graphs the lemma rules out.}
    \label{fig:my_label}
\end{figure}
\begin{proof}
There are two other possibilities:
\begin{enumerate}
    \item The equilibrium only involves lying (ie. $\psi(m_1)$ is neither $\theta_1$ nor $\theta_2$ and $\psi(m_2)\neq \theta_2$).
    \item $\psi(m_1)=\theta_2$.
\end{enumerate}
Both are easily verified to fail consistency:
$$
v_\text{LA}(m_1\rvert \theta_2)-v_\text{LA}(m_2\rvert \theta_2)\ge 0\ge v_\text{LA}(m_1\rvert \theta_1)-v_\text{LA}(m_2\rvert \theta_1).\qedhere
$$
\end{proof}
\begin{proof}[Proof of Proposition \ref{prop:LA}]
As observed in Section \ref{sec:bigeps},  $\SLIDS\subseteq \GCM$ in general, thus (a)$\Rightarrow$(b). Moreover, (c)$\Rightarrow$(a) is easily verified by computation. It remains to show that (b)$\Rightarrow$(c), ie. if $v_\text{LA}\in\GCM$ then $G(\sigma)$ is one of the graphs in Figure \ref{fig:LAgraph}.

We show the previous lemma rules out the subgraphs illustrated in Figures \ref{fig:imp1}, \ref{fig:imp2}. The impossibility of Figure \ref{fig:imp1} implies that for connected communication graphs $G$ with at most one revelation message (as implied by our restriction that the receiver has a unique best response to pure beliefs) there must exist a state that sends every message --- ie. the half-square $G^2[M]$ must be a complete graph. Since there is a state that sends every message, this state must neighbour every other state in $G^2[\Theta]$.

The impossibility of Figure \ref{fig:imp2} shows that if three states send a message, then every state sends that message --- knowing that there is at least one state that sends every message, this implies that $G^2[\Theta]$ is either a star or a complete graph.

The impossibility of these figures are obtained by repeated applications of Lemma \ref{lem:N}:

\noindent\textbf{Impossibility of Figure \ref{fig:imp1}:} Applying the lemma to the `N' spanned by $\theta_3$-$m_1$ implies that $\psi(m_2)\neq \theta_2$. 

\noindent With this knowledge, examining the `N' spanned by $\theta_2$-$m_3$, we find $\psi(m_3)=\theta_3$.

\noindent Revisiting $\theta_3$-$m_1$, knowing that $\psi(m_2)\neq \theta_3$, we deduce $\psi(m_1)= \theta_2$.

\noindent But this makes consistency on the `N' spanned by $\theta_1$-$m_2$ impossible.

\noindent\textbf{Impossibility of Figure \ref{fig:imp2}:} for the two `N' shapes spanned by $\theta_3$-$m_1$ and $\theta_4$-$m_1$ to hold, we must have $\psi(m_1)=\theta_2$ (since $\psi(m_2)$ cannot simultaneously be both $\theta_3$ and $\theta_4$). 

\noindent But as before this contradicts the `N' shape spanned by $\theta_1$-$m_2$. 
\end{proof}
\begin{proof}[Proof of Proposition \ref{prop:portfolio}]
\textbf{Single-Disclosure Case:}
    If $u^\ast(\overline{\theta})< u^\ast(\mu_{\{\overline{\theta}\}^c})$ then garbling the posterior $\mu_{\{\overline{\theta}\}^c}$ will decrease the utility until eventually we reach equality, creating a cheap talk candidate equilibrium with single-disclosure communication graph. Reversing this process shows this condition is necessary for a single-disclosure graph.
    
    In the second case, we combine garbling the posterior $\mu_{\{\overline{\theta}\}^c}$, and burning money for whichever posterior secures higher utility between $\overline{\theta}$ and $\mu_{\{\overline{\theta}\}^c}$ to obtain equilibrium. Again undoing this process reveals this condition is necessary.

\textbf{Partial-Revelation Case:}
    It is clear that a partial-revelation communication graph with $\theta^\ast=\theta_0$ cannot attain a utility above the next minimal state $\theta_1$. Indeed every such communication graph induces a posterior in $[\theta_0,\theta_1]$, and by quasiconvexity this cannot secure a utility greater than $u^\ast(\theta_1)$.

    Now consider a partial-revelation graph with $\theta^\ast=\overline{\theta}\neq \theta_0$. Note that such a graph must induce a posterior $\nu_0\in[\theta_0,\theta^\ast]$. Moreover, by varying the sender strategy, we can see that the posterior is in fact in the interval $[\theta_0,\mu_{\{\theta_0,\theta^\ast\}}]$. By the assumption of the proposition and quasiconvexity, this will result in a utility of at most $ u^\ast(\theta)$.
\end{proof}

\begin{proof}[Proof of Proposition \ref{prop:WD}]
Let the equilibrium messages be the complement of the support of the induced belief. Then for a path $(\theta_0,\pi_0,\dots,\theta_N,\pi_N)$, we have
\begin{align*}
    \begin{cases}
        v_\text{WD}(\pi_{i}\rvert \theta_{i})-v_\text{WD}(\pi_{i-1}\rvert \theta_{i})=0& i>0\\
        v_\text{WD}(\pi_{i}\rvert \theta_{i})-v_\text{WD}(\pi_{i-1}\rvert \theta_{i})>0& i=0.
    \end{cases}
\end{align*}
Thus graph-cyclic monotonicity is automatic.

If $\ell_{m,\theta}\equiv \ell_m$, then
\begin{align*}
    v_\text{WD}(\pi_0\rvert \theta_0)= v_\text{WD}(\pi_0\rvert \theta_1)\\
    \begin{cases}
        v_\text{WD}(\pi_N\rvert \theta_0)<v_\text{WD}(\pi_N\rvert \theta_1)& N=1\\
        v_\text{WD}(\pi_N\rvert \theta_0)=v_\text{WD}(\pi_N\rvert \theta_1)& N\ge 2,
    \end{cases}
\end{align*}
implying support monotonicity (since the modification is action independent).

\end{proof}
\section{Supplemental Material}\label{sec:ext}

\subsection{Money Burning Example}\label{app:burn}

We illustrate this technique in constructing equilibria of the money burning model:
\begin{ex}
Consider the message space $M=M_1\times M_2$ where $M_1$ denotes cheap talk messages, and $M_2=\mathbb{N}_0$ denotes a discrete quantity of money burnt that accompanies the message. Suppose the sender's utility takes the form
\begin{equation*}
    u(a,(m_1,m_2))=\hat{u}(a)-\frac{m_2}{N}
\end{equation*}
for some $N$ (representing currency increments). Note that assuming the currency space is discrete maintains Assumption (S).

In practice money burning might involve extravagant spending, wining and dining, or effort in building relationships with potential customers.

When $N$ is large, this model allows us to attain the maximum set of candidate equilibrium:
\begin{prop}\label{prop:burn}
    Suppose $|M_1|\ge |S|$, then for sufficiently large $N$, any Bayes-plausible set of posterior beliefs that induces at most one pure action corresponds to a burning-money equilibrium when the sender has transparent preferences.

    However these equilibria cannot improve the sender's utility beyond the best cheap talk equilibria.
\end{prop}
Note that this applies whenever persuasion is possible (ie. for some belief the receiver would prefer an action that is not a best response to the prior --- for generic $\mu$ this is equivalent to the receiver not possessing a dominant action across all states). In particular, this permits communication even when the sender has monotonic preferences over the receiver's beliefs.
\begin{proof}
    Let $\mathcal{B}$ be a plausible set of posterior beliefs, with $\mathcal{B}_0\subseteq \mathcal{B}$ the set of posterior beliefs inducing a mixed action. We show how to design scheme of money-burning that satisfies sender-incentive compatibility. Let $N$ be such that
    \begin{align*}
        \frac{2}{N}\le& \min_{\nu\in \mathcal{B}_0}\{\max \hat{u}^\ast(\nu)-\min \hat{u}^\ast(\nu)\}
    \end{align*}
    where $\hat{u}^\ast(\nu)$ is the range of utilities attained by a belief $\nu$ assuming no money is burnt. For convenience define $\underline{u}:=\min_{\nu\in\mathcal{B}}\{\max \hat{u}^\ast(\nu)\}$.
    
    If $\nu\in\mathcal{B}_0$ is such that $\min\hat{u}(\nu)\le \underline{u}-\frac{1}{N}$, we allow the belief to be induced without burning money. Otherwise, we require the sender to burn an amount $\frac{m_\nu}{N}$ so that
$$
\underline{u}-\tfrac{2}{N}< \min \hat{u}^\ast(\nu)-\tfrac{m_\nu}{N}\le\underline{u}-\tfrac{1}{N}
$$
Then $[\underline{u}-\frac{1}{N},\underline{u}]\subseteq \hat{u}^\ast(\nu)-\frac{m_\nu}{N}$ for all beliefs $\nu$ inducing a mixed action. For the message that induces the pure action $a$, require burning an amount of money $\frac{m}{N}$ so that $\hat{u}(a)-\frac{m}{N}\in [\underline{u}-\frac{1}{N},\underline{u}]$.

To see that this cannot improve the sender's utility beyond the best cheap talk case, suppose a money-burning equilibrium improves on the utility generated by the prior. Utility is maximized when there is at least one message where money isn't burnt, attaining the candidate equilibrium $\overline{u}$. But then every message involving the burning of money must lead to a strictly higher utility. By garbling these messages to move their posteriors towards the prior, the intermediate value theorem ensures that at some point the posterior have a best response leading to sender utility $\overline{u}$, meaning no money need be burnt to send this message.
\end{proof}   
\end{ex}
This construction subsumes the cheap talk equilibria as the special case where $\max_{\nu\in\mathcal{B}}\min \hat{u}^\ast(\nu)\le \min_{\nu\in\mathcal{B}}\max \hat{u}^\ast(\nu)$ and no burning of money is required.

An advantage of money-burning over cheap talk is that it allows persuasion when the sender's payoff is monotonic in the receiver's belief space. For example: a salesperson exists in two possible states: $s_a$ their product is good, $s_n$ their product is bad; and the receiver has two actions: $A$ buy the product, or $\emptyset$ buy nothing. Persuasive cheap talk is impossible in such a model, but money-burning can be quite persuasive in convincing the receiver to purchase the product. 

While many persuasion techniques observed in practice may be described as money burning, it is certainly not as ubiquitous as this reasoning would suggest. The fragility explored in Sections \ref{sec:frag} explains why money-burning may not always be credible, and Section \ref{sec:stab} offers a theoretical explanation as to when money-burning is a credible communication technology.

\subsection{Weak Stabilization}\label{app:weakstab}
We constrast the notion of $\mathscr{O}$-stability given in Definition \ref{def:stab} (which for clarity we will refer to as \emph{strong stability}) with the following weaker notion of stability:
\begin{defn}
    We say an equilibrium $\sigma$ is \emph{weakly-stabilized} by a modification $v$ if there exists Harsanyi stable equilibria $\sigma_\epsilon$ to the game with preference $u+\epsilon v$ such that $\sigma_\epsilon\rightarrow \sigma$.
\end{defn}
If we decompose our sender utility
\begin{equation*}
    u_S(p\rvert \theta,\omega)=u(p)+\epsilon v(p\rvert \theta)+ U_\epsilon (p\rvert \theta,\omega)
\end{equation*}
then strong stability requires that there exists an $N$ such that an equilibrium approximates $\sigma$ whenever $\|U_\epsilon\|_\infty<N\epsilon$ asymptotically. This inequality indicates the degree by which dependency must dominate idiosyncrasy for the desired communication to be an equilibrium.

Weak stability requires merely that there is some strictly positive function $f$ such that an equilibrium approximates $\sigma$ whenever $\|U_\epsilon\|_\infty<f(\epsilon)$ asymptotically. Thus weak-stability may demand that state-dependency dominate idiosyncrasy by an arbitrary degree of magnitude.


We use the following example to illustrate the fundamental difference between candidate equilibria that can only be weakly stabilized and those that may be strongly stabilized:

\begin{ex}[Weak Stabilization]\label{ex:weak}
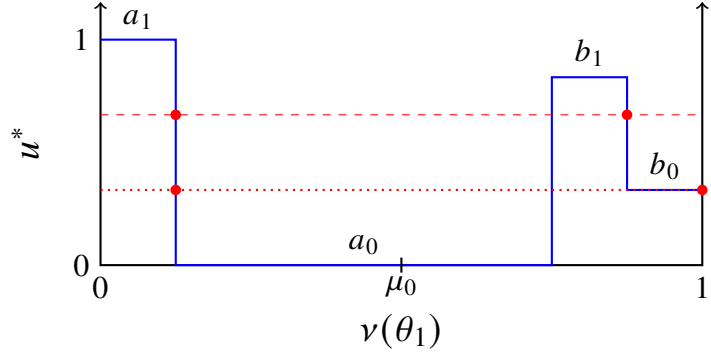
\begin{figure}
    \centering
    \begin{tikzpicture}
    \draw[thick,<->] (0,3.5) -- (0,0) -- (8,0) -- (8,3.5);
    \draw[thick, blue] (0,3) -- (1,3) -- (1,0) -- (6,0) -- (6,2.5) -- (7,2.5) -- (7,1)-- (8,1);
    \draw[red, dashed] (0,2) -- (8,2);
    \draw[red, dotted, thick] (0,1) -- (8,1);
    \draw[thick] (4,0.1) -- (4,-0.1);
    \node[anchor = east] at (0,0) {$0$};
    \node[anchor = east,rotate=90] at (-1,2) {\large$u^\ast$};
    \node[anchor = east] at (0,3) {$1$};
    \node[anchor = north] at (4,0) {$\mu_0$};
    \node[anchor = north] at (4,-0.5) {{\large$\nu(\theta_1)$}};
    \node[anchor = north] at (0,0) {$0$};
    \node[anchor = north] at (8,0) {$1$};
    \fill[red] (8,1) circle (2pt);
    \fill[red] (1,1) circle (2pt);
    \fill[red] (1,2) circle (2pt);
    \fill[red] (7,2) circle (2pt);
    \node[anchor = south] at (3.5,0) {$a_0$};
    \node[anchor = south] at (7.5,1) {$b_0$};
    \node[anchor = south] at (0.5,3) {$a_1$};
    \node[anchor = south] at (6.5,2.5) {$b_1$};
    \end{tikzpicture}
    \caption{The indirect utility considered in Example \ref{ex:weak}.}
    \label{fig:weaku}
\end{figure}
Consider the indirect utility illustrated in blue in Figure \ref{fig:weaku}, similar to Example 2 of \cite{SGGK23}. Utility will be normalized so that modifications only adjust the utility of $b_1,b_0$. There are three separate types of candidate equilibria we will analyze:

\begin{enumerate}
    \item The dotted red line illustrates the unique candidate equilibrium that can be strongly stabilized, which involves a message $m_b$ that induces action $b_0$, and a message $m_a$ that induces the appropriate degree of randomization $p_a$ between $[a_0,a_1]$. 
    \item There is also a range of candidate equilibria that can only be weakly stabilized ---- an example is given by the dashed red line --- that involve a message $m_b$ that induces a randomization $p_b$ strictly between $]b_0,b_1[$, and another message $m_a$ that induces the corresponding degree of randomization $p_a$ between $[a_0,a_1]$.
    \item There are additional candidate equilibria cannot be even weakly stabilized --- these are the candidate equilibria that involve the receiver randomizing between $[a_0,b_1]$ after one message, and $[a_0,a_1]$ or $[b_0,b_1]$ after another message.
\end{enumerate}
In what follows, we label message-action pairs by their action for simplicity.

\textbf{Candidate Equilibrium 1:} The first situation is analyzed in the general case of Section \ref{sec:stab} (as is the inability of the remaining equilibria to be strongly stabilized).

\textbf{Candidate Equilibria 2:} The second situation can be weakly stabilized in the following manner, illustrated in Figure \ref{fig:weakmod}: by shifting the utility of either $b_0,b_1$ --- in this case $b_1$ --- above the normalized line $[a_0,a_1]$ and moving the other below this line, we create an intersection point. This intersection represents a degree of randomization $p_a,p_b$ between the two pairs of actions that is necessary to make the sender in each state indifferent between the messages $m_a,m_b$. Observe that this modification can be interpreted as making one message riskier in one state compared with another (in this case message $m_b$ has a better best outcome and worse worst outcome in state $\theta_1$ than in state $\theta_0$).

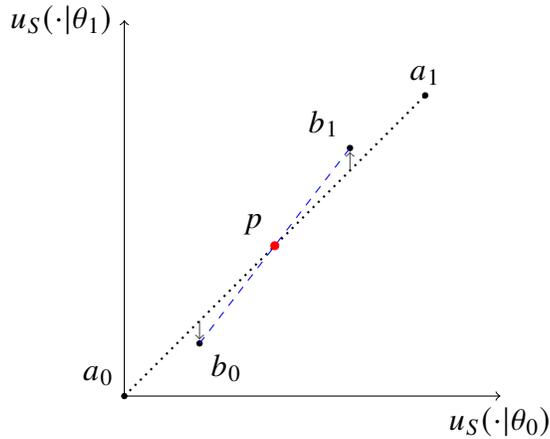
\begin{figure}
    \centering
    \begin{tikzpicture}
        \draw[->,black!70!white] (1,1) -- (1,0.75);
        \draw[->,black!70!white] (3,3) -- (3,3.25);
        \filldraw (4,4) circle (1 pt) node[anchor = south] {\color{black}$a_1$};
        \filldraw (0,0) circle (1 pt) node[anchor = south east] {\color{black}$a_0$};
        \filldraw (3,3.3) circle (1 pt) node[anchor = south east] {\color{black}$b_1$};
        \filldraw (1,0.7) circle (1 pt) node[anchor = north west] {\color{black}$b_0$};
        \draw[dotted, thick] (0,0) -- (4,4);
        \draw[dashed, blue] (1,0.7) -- (3,3.3);
        \node[anchor = east] at (0,5) {$u_S(\cdot \rvert \theta_1)$};
        \node[anchor = north] at (5,0) {$u_S(\cdot \rvert \theta_0)$};
        \filldraw[red] (2,2) circle (1.5 pt) node[anchor = south east] {\color{black}$p$};
        \draw[<->] (0,5) -- (0,0) -- (5,0);
    \end{tikzpicture}
    \caption{An illustration of a modification that can weakly stabilize the candidate equilibrium in Example \ref{ex:weak}. The axes are the utilities in the two states. The dotted black line indicates the utility of mixed actions between $[a_0,a_1]$ (normalized to be state independent), while the dashed blue line indicates the state-dependent preference of mixed actions $[b_0,b_1]$ (the modification is illustrated by the small grey arrows). 
    This creates a point of dual indifference at the intersection of the two lines, determining the degree of receiver randomization after each message (indicated in red).}
    \label{fig:weakmod}
\end{figure}

Given a small degree of sender idiosyncrasy, we can then locally manipulate the two variables $p_a,p_b$, until we get the appropriate degree of mixing from the senders in each state (a two-dimensional problem), to produce the desired posteriors.

The reason this equilibrium is only weakly stabilized is that this degree of randomization $p_a,p_b$ is highly sensitive to the modification. If we adjust the utility of $u_S(b_1\rvert \theta_0)$ slightly (holding everything else constant), the intersection, and thus degree of randomization, will shift a proportionate amount. In the strong stability case, we observe a similar effect when $u_S(b_0\rvert \theta_1)$ is adjusted, however this effect is the same (always a constant of proportionality of $1$), no matter how close we get to state-independence; while in the weak stability case the effect is exacerbated when we are close to state-independence, where we are looking for the intersection of nearly incident lines.

It is also clear why this process does not work to weakly stabilize either of the similar candidate equilibria involving pure actions ($b_1$ or $b_0$): to create a point of dual indifference for this equilibrium requires that the point $b_1$ is incident with the line $[a_0,a_1]$. But then the entire line $[b_0,b_1]$ lies weakly below the line, and subject to a small nudge to the utility of $b_1$ the lines do not intersect at all, and no equilibria with the required posteriors can be maintained.

\textbf{Candidate Equilibria 3:} The third family of candidate equilibria cannot even be weakly stabilized. This is because, in their situation, weak stabilization requires creating dual indifference between two mixed actions whose support only differ in one action. This is illustrated in Figure \ref{fig:noweakmod}. To create dual indifference it is necessary that $b_1$ lies precisely on the normalized (state-independent) line $[a_0,a_1]$. That is, =such communication requires state-independence, which we already know is fragile to idiosyncrasy.
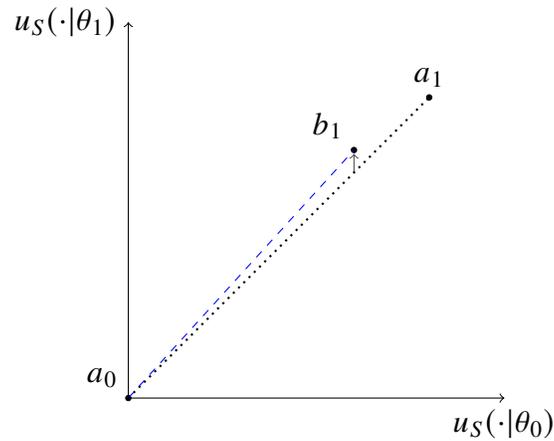
\begin{figure}
    \centering
    \begin{tikzpicture}
        \draw[->,black!70!white] (3,3) -- (3,3.25);
        \filldraw (4,4) circle (1 pt) node[anchor = south] {\color{black}$a_1$};
        \filldraw (0,0) circle (1 pt) node[anchor = south east] {\color{black}$a_0$};
        \filldraw (3,3.3) circle (1 pt) node[anchor = south east] {\color{black}$b_1$};
        \draw[dotted, thick] (0,0) -- (4,4);
        \draw[dashed, blue] (0,0) -- (3,3.3);
        \node[anchor = east] at (0,5) {$u_S(\cdot \rvert \theta_1)$};
        \node[anchor = north] at (5,0) {$u_S(\cdot \rvert \theta_0)$};
        \draw[<->] (0,5) -- (0,0) -- (5,0);
    \end{tikzpicture}
    \caption{An illustration of how the third type of equilibria in Example \ref{ex:weak} require state independence to create a point of dual indifference, and hence are inherently unstable. The black dashed line is the state-independent line (determined by normalization), which $u_S(b_1\rvert s)$ must lie on to create the point of dual indifference. Otherwise (as in the dashed blue line) they only intersect at the same action ($a_0$), and the action is message independent (ie. no persuasion occurs).}
    \label{fig:noweakmod}
\end{figure}
\end{ex}
\paragraph{Four types of complexity}
This example illustrates many of the complexities required and demanded by weak (but not strong) stabilization: (1) a higher degree of complexity in the equilibrium message structure, (2) a complex receiver environment, (3) complex modifications, and (4) (even small) uncertainty plays a role in which equilibria can be stabilized:

\textbf{(1) Message Structure:} 
The presence of cyclicity in the message structure demands more complex indifferences from the sender's perspective (which could be interpreted as more strategic complexity) as well as being in a less informative family of equilibria from the receiver's perspective (acyclic equilibria can be seen as a more informative class). In the above example the strongly stabilized equilibrium is strictly more Blackwell-informative than the weakly stabilized equilibrium.

\textbf{(2) Receiver environment:} 
As the third family shows, not all candidate equilibria can be weakly stabilized: weak stabilization requires the presence of additional actions, and will not differ from strong stability for equilibria where a single pure action is in the support of all actions.

\textbf{(3) Modification:} 
Furthermore, the reason these additional actions are necessary is that the required modifications are slightly more complex: they require adjusting the `riskiness' of a message in a state dependent way. This can be constrasted with strong stabilization where only the state-dependent \emph{attractiveness} of equilibrium messages needs to modified. To observe the complexity of the required modifications, note that action-independent modifications can never produce this effect on riskiness.\footnote{Empathy is also incapable of producing this effect in two-state environments, but may with more states.}

\textbf{(4) Uncertainty:} 
Lastly, this necessity of adjusting the riskiness means that pure actions are typically not present in weakly stabilized equilibria (as shown with the second family).
\vspace{2 pt}

These points do not show that weak stabilization is impossible, but rather that it is distinct from strong stabilization and, in some informal sense, demands more from both the environment and the parties involved.

\subsection{Relations with Mechanism Design}\label{app:CM}
\paragraph{Cyclic Monotonicity $\Leftrightarrow$ Graph Monotonicity}
In mechanism design there exists a concept of cyclic monotonicity that is closely related to our notion of graph monotonicity. The specific variation of cyclic monotonicity most relevant to our problem is the following:
\begin{defn}[$\sigma$-Cyclic Monotonicity]
For a candidate equilibrium $\sigma$,
we say that a modification $v$ is $\sigma$\emph{-cyclically monotone} (or $v\in \CM$) if, for any sequence $(\theta_1,\dots,\theta_{N-1},\theta_N=:\theta_{0})$ with $N\ge 1$, and any $\pi_i\in \sigma(\theta_i)$ we have
\begin{align}\label{eq:CM}\tag{CM}
\sum_{i=1}^{N} v(\pi_{i}\rvert \theta_{i})-v(\pi_{i-1}\rvert \theta_{i})>0.
\end{align}
\end{defn}
This differs from the graphical notion in that the sequence does not have to follow chains of indifferences. As a result this is a stronger condition. However in the situations that we consider in this paper, the notions coincide:
\begin{prop}\label{prop:coinc}
    In general $\CM\subseteq\GCM$, with equality iff one of the following holds:
    \begin{enumerate}
        \item $G(\sigma)$ is cyclic, in which case the sets are empty.
        \item $G(\sigma)$ is a tree, in which case the sets are non-empty.
    \end{enumerate}
\end{prop}
The first property immediately follows from $\GCM $ being empty in such situations, as shown in Proposition \ref{prop:acyc}.

The second fact results from the inequalities of eq. \ref{eq:GCM} forming a basis for the inequalities of eq. \ref{eq:CM} when $G(\sigma)$ is connected.

A slight modification to our proof of Theorem \ref{thm:main} will show that cyclic monotonicity is a necessary condition for sender $\mathscr{O}$-stability, as such it is a `tighter' condition than graph monotonicity. However, in light of Theorem \ref{thm:gen} showing that the situations where these sets differ are fragile and rare, we opt for the simpler concept, which is more relevant to our proof strategies. Furthermore, strong graph-montonicity and support strong monotonicity have little relation to cyclic monotonicity in forest graphs --- it is not clear how one should adjust the usual concept of strong monotonicity from mechanism design to allow for indifferences.

\paragraph{Weak Monotonicity and Strong Graph Monotonicity}
In mechanism design weak monotonicity provides an easy-to-verify necessary condition for incentive compatibility:
\begin{defn}[Weak monotonicity]
    For a candidate equilibrium $\sigma$, we say that a modification $v$ is $\sigma$\emph{-weak monotone} (or $v\in \text{WM}(\sigma)$) if, for any $\theta,\theta'\in\Theta$, and any $\pi\in \sigma(\theta),\pi'\in\sigma(\theta')$ we have
\begin{align}\label{eq:WM}\tag{WM}
v(\pi\rvert \theta)-v(\pi'\rvert \theta)>v(\pi\rvert \theta')-v(\pi'\rvert \theta').
\end{align}
\end{defn}
\begin{prop}
In general $\SLID\subseteq \GCM\subseteq\text{WM}(\sigma)$. However if $\sigma$ is a quasi-interval candidate equilibrium then $\SLID=\GCM=\text{WM}(\sigma)$.
\end{prop}
\end{document}